\documentclass[journal]{IEEEtran}

\usepackage{graphicx} 
\usepackage{caption}
\usepackage{epstopdf}
\usepackage{subfigure}
\usepackage{caption}
\usepackage{subcaption}
\usepackage{stfloats}
\usepackage{algorithm, algorithmic}
\usepackage{diagbox}
\usepackage{multirow}
\usepackage{mathtools}
\usepackage{setspace}  
\usepackage{threeparttable}
\usepackage{textcomp,booktabs}
\usepackage[usenames,dvipsnames]{color}
\usepackage{colortbl}
\usepackage{indentfirst}
\usepackage{cite}
\usepackage{amsmath,amssymb,amsfonts}
\usepackage{algorithmic}
\usepackage{graphicx}
\usepackage{textcomp}
\usepackage{xcolor}
\usepackage{mathtools}
\usepackage{amsthm}
\usepackage{color}
\usepackage{enumitem}
\usepackage[utf8]{inputenc}

\newtheorem{lemma}{\bf Lemma}

\theoremstyle{definition}

\newtheorem{theorem}{Theorem}

\definecolor{mygray}{gray}{.9}
\definecolor{mypink}{rgb}{.99,.91,.95}
\definecolor{mycyan}{cmyk}{.3,0,0,0}

\begin{document}

\title{Unveiling the Power of Complex-Valued Transformers in Wireless Communications\\
}

\author{\IEEEauthorblockN{Yang Leng, Qingfeng Lin, Long-Yin Yung, Jingreng Lei, Yang Li, and Yik-Chung Wu}
\thanks{Yang Leng, Qingfeng Lin, Long-Yin Yung, Jingreng Lei, and Yik-Chung Wu are with the Department of Electrical and Electronic Engineering, The University of Hong Kong, Hong
Kong (e-mail: \{lengyang, qflin, lyyung, leijr, ycwu\}@eee.hku.hk). Qingfeng Lin and Jingreng Lei are also with Shenzhen Research Institute of Big Data, Shenzhen~518172, China.}
\thanks{ Yang Li is with Shenzhen Research Institute of Big Data,
Shenzhen~518172, China, and also with the School of Science and Engineering, The Chinese University of Hong Kong, Shenzhen 518172, China (e-mail: liyang@sribd.cn).}
}

\maketitle

\begin{abstract} Utilizing complex-valued neural networks (CVNNs) in wireless communication tasks has received growing attention for their ability to provide natural and effective representation of complex-valued signals and data. However, existing studies typically employ complex-valued versions of simple neural network architectures. Not only they merely scratch the surface of the extensive range of modern deep learning techniques, theoretical understanding of the superior performance of CVNNs is missing. To this end, this paper aims to fill both the theoretical and practice gap of employing CVNNs in wireless communications. In particular, we provide a comprehensive description on the various operations in CVNNs and theoretically prove that the CVNN requires fewer layers than the real-valued counterpart to achieve a given approximation error of a continuous complex-valued function. Furthermore, to advance CVNNs in the field of wireless communications, this paper focuses on the transformer model, which represents a more sophisticated deep learning architecture and has been shown to have excellent performance in wireless communications but only in its real-valued form. In this aspect, we propose a fundamental paradigm of complex-valued transformers for wireless communications, including the complex-valued embedding module, encoding module, decoding module, and output projection module. Leveraging this structure, we develop customized complex-valued transformers for three representative applications in wireless communications: channel estimation, user activity detection, and joint design of pilot, feedback quantization, and precoder. These applications utilize transformers with varying levels of sophistication and span a variety of tasks, ranging from regression to classification, supervised to unsupervised learning, and specific module design to end-to-end design. Experimental results demonstrate the superior performance of the complex-valued transformers for the above three applications compared to other traditional real-valued neural network-based methods.

\end{abstract}


\section{Introduction}
The emergence of neural networks has brought a paradigm shift in the field of wireless communications~\cite{zhou2020service}, offering a novel approach to address the challenges posed by model-mismatch in knowledge-driven methods~\cite{peng2022computing, wang2023engnn, li2024hpe}. A variety of deep learning architectures, such as multilayer perceptrons (MLPs), convolution neural networks (CNNs), and transformers, have demonstrated remarkable success in learning intricate patterns and representations from data. In particular, the transformer model distinguishes itself by its superior capability to extract the correlation among high-dimensional signals. However, these deep learning models are based on real-valued operations and often struggle to effectively capture the inherent complex-valued property of wireless signals.

Promising better performance for complex-valued domain tasks, complex-valued neural networks (CVNNs) have received growing attention for their ability to directly process complex-valued data~\cite{MRI, SAR, complex_ofdm}. On one hand, they can capture phase variations and amplitude fluctuations of complex-valued signals that are imperceptible to real-valued neural networks (RVNNs). This property is particularly important in wireless communications since the inputs and outputs of many wireless communication tasks are naturally complex-valued, e.g., constellation points and channel responses. \textcolor{black}{On the other hand, it is theoretically proved that the estimation
error bound of the CVNN is smaller than its real-valued counterpart so that it can bring a further performance gain~\cite{CVTNN}.
These make} CVNNs particularly well-suited for representing and processing data from wireless communications.

\textcolor{black}{In fact, CVNNs have demonstrated their advantages in various wireless communication applications recently~\cite{cvdlwc_survey}. For example, in frequency-division duplexing (FDD) systems, a complex-valued convolution neural network (CVCNN) was proposed in~\cite{zhang2020CVFDD} to leverage the signal structures across subcarriers and antennas for enhancing downlink channel estimation based on uplink measurements. Moreover, the CVCNN proposed in~\cite{marseet2017application} achieves lower bit error rates than traditional maximum likelihood detectors while reducing computational complexity in the context of symbol detection in multiple input multiple output (MIMO) systems. Similarly, in modulation classification task, CVCNNs accurately classify different modulation schemes by capturing the phase and magnitude variations inherent in the modulated signals~\cite{CVCNN_MR1, CVCNN_MR2, CVCNN_MR3}. As an end-to-end approach for hybrid precoding, a complex-valued MLP is leveraged in~\cite{alrabeiah2019TDDFDD} to adapt the compressive channel sensing vector to the environment and predict the precoding coefficients.} 
Further examples of CVNNs in wireless communications include cognitive radio~\cite{CVNN_radio}, signal recognition~\cite{signaltransformer}, and millimeter wave communication systems~\cite{CVNN_mmwave}, where the complex-valued nature of the transmitted signals are effectively exploited by CVNNs. 

\textcolor{black}{Despite the demonstrated advantages of CVNNs in wireless communication tasks, most current explorations utilize basic architectures without attention mechanisms~\cite{zhiyan, resOFDM, csinet2018deep, bai2019CEfading, csinet2021adaptive}. Incorporating attention mechanisms through transformer architectures addresses this limitation by capturing the varying relevance among input features. As attention mechanism is inherently a process of finding correlations among features derived from the input data, it reminisces many operations in a wireless system. Therefore, real-valued transformer has been recently demonstrated to offer significant benefits in communication tasks such as channel estimation~\cite{luan2023channelformer}, activity detection~\cite{LY_transformer}, and precoding in FDD systems~\cite{jiang2023transformer}. These examples highlight the importance of transformer structures in communication applications.}

\textcolor{black}{Given the benefits of CVNNs and transformer architecture, it is natural to contemplate a complex-valued transformer for wireless communications. The key step in transformer is the attention mechanism. Original transformers are based on real-valued data, so defining complex-valued attention is the first challenge. Some studies \cite{dlsignal,signaltransformer} extend the attention mechanism directly to the complex-valued domain but overlook the fundamental principle of measuring similarity through dot-product. To this end, \cite{2023building} introduces a novel complex-valued scaled dot-product attention mechanism, which measures the similarity of complex-valued vectors by cosine similarity. However, the complex-valued transformer with their proposed attention mechanism has so far only been applied to the task of automatic music transcription. The lack of a unified and comprehensive framework hinders further exploration of more complex tasks in wireless communications and other fields.}

\textcolor{black}{In this paper, we present a general framework for incorporating the complex-valued transformers, aiming to provide guidance for easy implementation in the context of wireless communications. We first present a comprehensive exposition of the CVNN, serving as various building blocks for the more advanced complex-valued transformer. We also establish the theoretical advantage of CVNN over RVNN in terms of the number of layers to achieve the same functional approximation error. Then, the general architecture of complex-valued transformer and various modules are detailed. To demonstrate the effectiveness of complex-valued transformer in wireless communications, we apply it to three representative applications: channel estimation, activity detection, and precoding design in FDD systems, covering scenarios ranging from regression to classification, supervised to unsupervised learning, and specific module design to end-to-end design. The complex-valued transformers we employed in these applications are not directly derived from \cite{2023building}, and they vary in degree of sophistication, ranging from standard transformer to heterogeneous transformer, and even with cross-branch attention between encoder and decoder in the transformer. To the best of our knowledge, this is the first work demonstrating the potential and merit of complex-valued transformers in wireless communications.}

The remainder of this paper is organized as follows. The complex-valued neural network operations are inspected in Section~\ref{Sec_cvnn}.  Section~\ref{sec_theory} presents the theoretical grounding of CVNN over RVNN. Section~\ref{sec_cvtrans} introduces the details of the complex-valued transformer, including the general architecture, the complex-valued attention mechanism, and multi-head attention. Based on this architecture, Sections~\ref{case1} to~\ref{case3} adopt the complex-valued transformer to three representative tasks in wireless communication as case studies. Numerical results are provided in Section~\ref{Sec_Results}, followed by the conclusions in Section~\ref{Sec_Conclusion}.

\section{Complex-Valued Neural Network Operations} \label{Sec_cvnn}

\begin{figure}[tb] 
	\centering
	\subfigure[Linear Module Type 1]{ 
		\label{Linear1}
		\includegraphics[width=3.5in,height=1.3in]{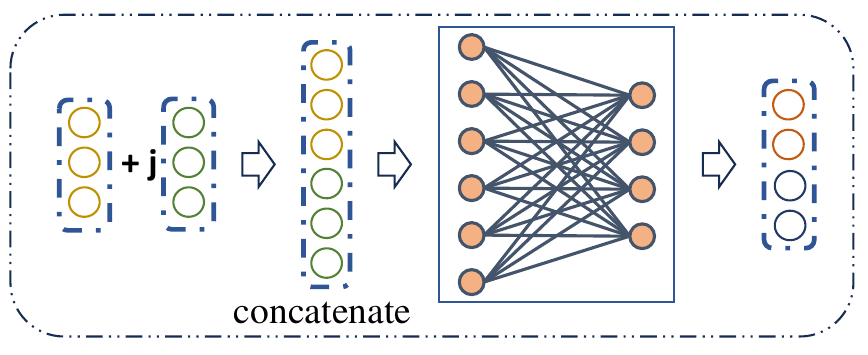}} \hspace{0in}
  \subfigure[Linear Module Type 2]{ 
		\label{Linear2}
		\includegraphics[width=3.5in,height=1.3in]{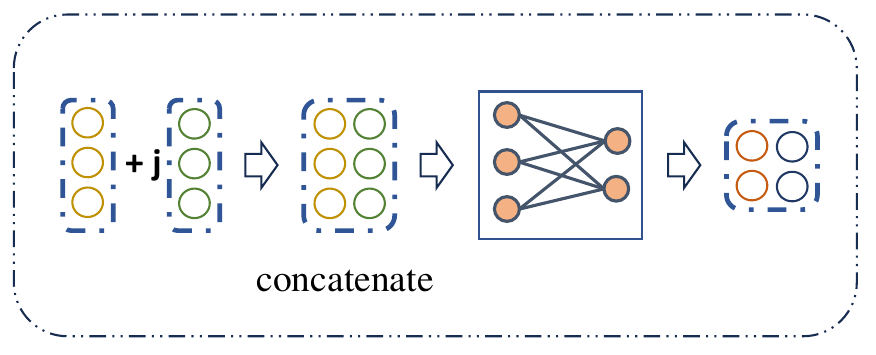}} \hspace{0in}
    \subfigure[$\mathbb{C}$Linear Module]{ 
		 \label{CLinear}
      \includegraphics[width=3.6in,height=1.3in]{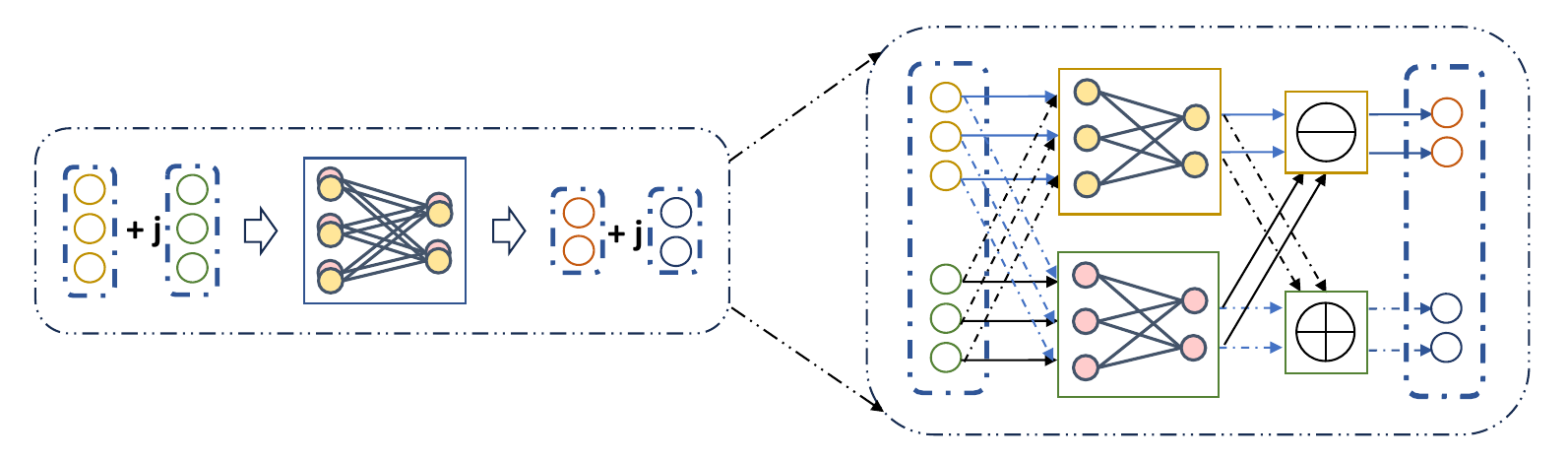}} \hspace{0in}
	\caption{The internal mechanism of Linear and $\mathbb{C}$Linear modules.} 
\end{figure}
\subsection{Processing Complex-Valued Signal in Neural Networks}
The linear transformation is a widely employed module in neural network architectures, which involves matrix multiplication of the input data with a weight matrix, followed by the addition of a bias term. This process allows the input data to be projected into spaces with different dimensions. 
A standard real-valued linear transformation  for a real-valued input $\mathbf{x} \in \mathbb{R}^{L}$ is denoted by:
 \begin{align}
\mathbf{y}=\mathbf{W} \mathbf{x}+\mathbf{b},
 \end{align}
where $\mathbf{y} \in \mathbb{R}^{d}$ is the projection output from dimension $L$ to $d$, and   $\mathbf{W} \in \mathbb{R}^{d\times L}$ and $\mathbf{b} \in \mathbb{R}^{d}$ represent the weight and bias, respectively.
To apply the linear transformation on complex-valued input signals, existing works based on RVNNs often split the signal into real and imaginary parts and concatenate them as the input of a real-valued linear transformation~\cite{luan2023channelformer, LY_transformer, jiang2023transformer}. 

In particular, for a complex-valued input signal $\mathbf{x}^{\mathbb{C}}\in \mathbb{C}^{L}$, the concatenation for the real and imaginary parts of $\mathbf{x}^{\mathbb{C}}$ can be performed in both axes, as shown in Fig.~\ref{Linear1} and Fig.~\ref{Linear2}, and they are represented as 
\begin{equation}\label{eq_concat}
    {\mathbf{x}^{\mathbb{R}} }=\begin{cases}
        \left[\Re\left(\mathbf{x}^{\mathbb{C}}\right)^T \, \Im\left(\mathbf{x}^{\mathbb{C}}\right)^T\right]^T \in \mathbb{R}^{2L},\quad&\text{Type 1},\\
        \left[\Re\left(\mathbf{x}^{\mathbb{C}}\right) \, \Im\left(\mathbf{x}^{\mathbb{C}}\right)\right] \in \mathbb{R}^{L\times 2},&\text{Type 2},
    \end{cases}
\end{equation}
respectively. To project the input into a $d$ dimensional complex-valued space, the first linear transformation in Fig.~\ref{Linear1} requires a weight $\mathbf{W} \in \mathbb{R}^{2d\times 2L}$ and a bias $\mathbf{b} \in \mathbb{R}^{2d}$, which costs $2d(2L + 1)$ real trainable scalars in total,   while the real-valued operation in Fig.~\ref{Linear2} requires only $d(L + 1)$ parameters. However, these two linear transformations in RVNNs both suffer from the inherent disadvantage of ignoring the structural correlation between real and imaginary parts of the complex-valued features.

The fully complex-valued linear operation, denoted as $\mathbb{C}$Linear and illustrated in Fig.~\ref{CLinear}, on the other hand, exploits the relationship between the real and imaginary parts as implicit constraints, which is represented as
\begin{equation}\label{ComplexLinear}
   \begin{split}
       \Re(\mathbf{y}^{\mathbb{C}}) + j \Im(\mathbf{y}^{\mathbb{C}}) 
=&(\Re(\mathbf{W}^{\mathbb{C}})  + j\Im(\mathbf{W}^{\mathbb{C}}) )(\Re(\mathbf{x}^{\mathbb{C}})\\&+j\Im(\mathbf{x}^{\mathbb{C}}))
+ \Re(\mathbf{b}^{\mathbb{C}}) +j\Im(\mathbf{b}^{\mathbb{C}}),
   \end{split} 
\end{equation}
where $\mathbf{W^{\mathbb{C}}} \in \mathbb{C}^{d\times L}$ and $\mathbf{b^{\mathbb{C}}} \in \mathbb{C}^{d}$ are complex-valued trainable parameters. Equation \eqref{ComplexLinear} can be further expressed as
\begin{small}
\begin{align} 
\left[\begin{array}{c}
\Re(\mathbf{y}^{\mathbb{C}})\\
\Im(\mathbf{y}^{\mathbb{C}})
\end{array}\right] = \left[\begin{array}{cc}
\Re(\mathbf{W}^{\mathbb{C}}) & 
-\Im(\mathbf{W}^{\mathbb{C}}) \\
\Im(\mathbf{W}^{\mathbb{C}}) & \Re(\mathbf{W}^{\mathbb{C}})
\end{array}\right] \left[\begin{array}{c}
\Re(\mathbf{x}^{\mathbb{C}}) \\
\Im(\mathbf{x}^{\mathbb{C}})
\end{array}\right] + \left[\begin{array}{c}
\Re(\mathbf{b}^{\mathbb{C}}) \\
\Im(\mathbf{b}^{\mathbb{C}})
\end{array}\right]. \nonumber
\end{align}
\end{small}It can be observed that the $\mathbb{C}$Linear operation only needs $d(L+1)$ complex-valued trainable scalars, or $2d(L+1)$ real-valued trainable scalars for each $\mathbf{x}^{\mathbb{C}}$.
The relationship between the real and imaginary parts is implicitly imposed by using relatively fewer trainable parameters, thus offering a more natural and effective framework for modeling and processing wireless signals.

\textcolor{black}{By cascading $L$ layers of complex-valued linear transformations, each followed by a nonlinear activation function, we construct a complex-valued fully connected network, denoted as $\mathbb{C}\operatorname{FCN}^{(L)}$. The network operates as:
\begin{align} \label{eq:CFCN} \mathbf{Y}^{(L)} = \mathbb{C}\operatorname{FCN}^{(L)}\left( \mathbf{Y}^{(0)} \right), \end{align}
with the intermediate layers defined recursively by:
\begin{equation}
 \mathbf{Y}^{(i)} = \begin{cases}\sigma^{(i)} \left( \mathbf{W}^{(i)} \mathbf{Y}^{(i-1)} + \mathbf{b}^{(i)} \right), & i = 1, 2, \dots, L-1, \\ \mathbf{W}^{(i)} \mathbf{Y}^{(i-1)} + \mathbf{b}^{(i)}, & i = L,\end{cases}  
\end{equation}
where $\mathbf{Y}^{(0)}$ represents the input of $\mathbb{C}\operatorname{FCN}^{(L)}$,  $\mathbf{Y}^{(L)}$ denotes the output, the $\mathbf{W}^{(i)}$ and $\mathbf{b}^{(i)}$ are the weights and biases of the $i$-th layer of the network, respectively. The $\sigma^{(i)}(\cdot)$ is the nonlinear activation function of the $i$-th layer network, which is generally realized by the operation $\mathbb{C}\operatorname{ReLU}(\cdot)$:
\begin{align}
\label{CReLU}
\mathbb{C} \operatorname{ReLU} \left(\cdot\right) =  \operatorname{ReLU}  \left(\Re(\cdot)\right) + i \operatorname{ReLU}  \left(\Im(\cdot)\right).
\end{align}}

Convolution is another fundamental module which is capable of capturing spatial information. The complex-valued convolution ($\mathbb{C}$Conv), which utilizes complex-valued filters to capture both the magnitude and phase information of complex-valued input, also has a similar internal mechanism of $\mathbb{C}$Linear. For example, in the case of two-dimensional convolution without bias, let $\mathbf{K}=(\Re(\mathbf{K})+ j\Im(\mathbf{K}) ) \in \mathbb{C}^{H_k\times W_k}$ represent a filter with size $H_k\times W_k$, the  $\mathbb{C}$Conv operation over the input $\mathbf{X}=(\Re(\mathbf{X})+ j\Im(\mathbf{X}) ) \in \mathbb{C}^{H\times W}$ can be expressed as 
% \begin{align\
\begin{equation}
    \begin{split}
        \mathbf{X} * \mathbf{K} &=(\Re(\mathbf{X})+j \Im(\mathbf{X})) *(\Re(\mathbf{K})+j \Im(\mathbf{K}))\\
&=\Re(\mathbf{X}) * \Re(\mathbf{K})-\Im(\mathbf{X}) * \Im(\mathbf{K})\\
&\quad\,\,\,\,+j(\Re(\mathbf{X}) * \Im(\mathbf{K})+\Im(\mathbf{X}) * \Re(\mathbf{K})).
    \end{split}
\end{equation}

\subsection{Complex-Valued Layer Normalization ($\mathbb{C}$LN)}
\label{CLN}
Layer normalization is crucial in ensuring that the model stabilizes forward propagation and mitigates the problem of vanishing or exploding gradients in back-propagation. 
Define $\mathbf{1}_d$ as a length-$d$ column vector with all its elements being ones. For a $d$-dimensional complex-valued vector $\mathbf{x}$, the mean and covariance matrix of the input features are calculated as
\begin{align}
\mu_{\mathbf{x}}\triangleq\frac{1}{d} \sum_{i=1}^{d} x_{i},
%, ~~\forall n = 1,\dots, N+1,
\label{eq_CLN1}
\end{align}
\begin{align}
&\mathbf{K}_{\mathbf{x}}\triangleq\left[\begin{array}{cc}
\operatorname{Var}(\Re(\mathbf{x})) & \operatorname{Cov}(\Re(\mathbf{x}), \Im(\mathbf{x})) \\
\operatorname{Cov}(\Re(\mathbf{x}), \Im(\mathbf{x})) & \operatorname{Var}(\Im(\mathbf{x}))
\end{array}\right],
\end{align}
where $\operatorname{Var}(\Re(\mathbf{x})) \triangleq \frac{1}{d}(\Re(\mathbf{x})-\Re(\mu_{\mathbf{x}})\mathbf{1}_d)^T(\Re({\mathbf{x}})-\Re(\mu_{\mathbf{x}})\mathbf{1}_d)$ and $\operatorname{Cov}(\Re(\mathbf{x}),\Im(\mathbf{x}))\triangleq\frac{1}{d}\mathbb(\Re(\mathbf{x})-\Re(\mu_{\mathbf{x}})\mathbf{1}_d)^T(\Im(\mathbf{x})-\Im(\mu_{\mathbf{x}})\mathbf{1}_d)$. Then, the normalization results are given by
\begin{small}

\begin{align} \label{eq_CLN2}
&\left[\begin{array}{l}
\Re(\mathbb{C} \operatorname{LN}(\mathbf{x}))^T\\
\Im(\mathbb{C}\operatorname{LN}(\mathbf{x}))^T
\end{array}\right]=\boldsymbol{\Lambda}^{\frac{1}{2}} \mathbf{K}_{\mathbf{x}}^{-\frac{1}{2}}\left[\begin{array}{l}
\Re(\mathbf{x}-\mu_{\mathbf{x}}\mathbf{1}_d)^T \\
\Im(\mathbf{x}-\mu_{\mathbf{x}}\mathbf{1}_d)^T
\end{array}\right]+\boldsymbol{\beta}\left[\begin{array}{l}
\mathbf{1}_d^T \\
\mathbf{1}_d^T
\end{array}\right],
\end{align}
\end{small}
where $\boldsymbol{\beta}\in\mathbb{C}^2$ and $\boldsymbol{\Lambda} \in \mathbb{S}_{+}^{2}$ are trainable parameters. \textcolor{black}{The term $\mathbf{K}_{\mathbf{x}}^{-\frac{1}{2}}$ incorporates the correlation between real and imaginary parts of $\mathbf{x}$. When its off-diagonal elements are zero,  \eqref{eq_CLN2} reduces to the standard layer normalization operation.}

\subsection{$\mathbb{C}2\mathbb{R}$ Layer}
\label{C2R}

The $\mathbb{C}2\mathbb{R}$ layer aims to project complex-valued features to real values before applying the loss function for the classification tasks. Specifically, the final output of a classification task should be a real-valued probability vector, which conflicts with the nature of CVNNs. To effectively link the final layer of a CVNN with the cross entropy loss function, we introduce the $\mathbb{C}2\mathbb{R}$ layer below for the two-class activity detection task. For multi-class classification, this can be extended by replacing the activation function accordingly. 

Let $\mathbf{x}^\mathbb{C}$ be the complex-valued output neuron representing class information. We adopt the Type 1 concatenation strategy of~\eqref{eq_concat} to represent complex-valued features  $\mathbf{x}^\mathbb{R}=[\Re(\mathbf{x}^\mathbb{C}),\Im(\mathbf{x}^\mathbb{C})]$. 
Then, a real-valued linear transformation is used to aggregate the two parts as the input of the $Sigmoid$ function, which can be expressed as
\begin{align}
p &= \mathbb{C}2\mathbb{R}(\mathbf{x}^{\mathbb{C}}) = \text{Sigmoid}(\mathbf{w}^T_\text{out}{\mathbf{x}^{\mathbb{R}}}^T+{b}_\text{out}), 
\label{eq_C2R}
\end{align}
where the output $p$ is a real value between 0 and 1, representing the probability of the first class, $\mathbf{w}_\text{out} \in \mathbb{R}^{2}$ and ${b}_\text{out} \in \mathbb{R}$ are weight and bias of the real-valued linear layer.

\subsection{Complex-Valued Back-Propagation}
\textcolor{black}{Back-propagation is a key procedure for training neural networks. In CVNNs, back-propagation extends to complex parameters by employing the concept of Wirtinger derivatives, which can be implemented by Pytorch. Unlike real-valued functions, complex functions may not be holomorphic, and thus the standard rules of complex differentiation do not apply. To address this, we treat the loss function $L$ as a function of both the complex variable $z$ and its complex conjugate $z^*$. The Wirtinger derivatives are defined as \cite{wirtinger1927formalen}:
\begin{align}
\frac{\partial L}{\partial z}=\frac{1}{2}\left(\frac{\partial L}{\partial x}-i \frac{\partial L}{\partial y}\right), \quad \frac{\partial L}{\partial z^*}=\frac{1}{2}\left(\frac{\partial L}{\partial x}+i \frac{\partial L}{\partial y}\right),
\end{align}
where $z=x+i y$, with $x$ and $y$ being the real and imaginary parts of $z$, respectively. The update rule for the complex parameter $z$ during backpropagation is then given by \cite{kreutz2009complex}:
\begin{align}
z \leftarrow z-\eta \frac{\partial L}{\partial z^*},
\end{align}
where $\eta \in \mathbb{R}$ is the learning rate. Notably, the gradient with respect to $z^*$ is used in the update, ensuring the loss function decreases along the steepest descent direction in the complex plane  \cite{kreutz2009complex}. This mechanism allows for the effective training of complex-valued neural networks by extending the principles of gradient descent to complex domains, even when the functions involved are not holomorphic.}
\section{Theoretical Advantage of CVNN}
\label{sec_theory}
One may wonder what is the advantage of using CVNN instead of RVNN. The following lemma characterizes the depth of CVNNs required to approximate a complex-valued mapping $f$ with error of at most $\varepsilon$. 
\begin{lemma}\label{lemma: cvnn_bound}
    Let $f:[-1,1]^n+i[-1,1]^n\rightarrow \mathbb{C}^m$ be a $K$-Lipschitz continuous function and $\varepsilon>0$ be given. Then, there exists a CVNN $\Phi^\mathbb{C}$ with width $2n+2m+1$ and depth $N^\mathbb{C}$ such that $\Vert f-\Phi^\mathbb{C} \Vert\leq \varepsilon$, where the depth $N^\mathbb{C}$ is of order
    \begin{equation}
        N^\mathbb{C}\sim \mathcal{O}\left(\left(\frac{144mK^2(n+2)^2}{\varepsilon^2}+9\right)^{2n}\right).
    \label{lemma1}
    \end{equation}
    
\end{lemma}
\begin{proof}
Based on~\cite[ Theorem 5.3]{theory1}, the depth $N^\mathbb{C}$ of CVNN $\Phi^\mathbb{C}$ such that $|| f- \Phi^\mathbb{C}|| \leq \epsilon $ is upper bounded by
\begin{equation}\label{eqn: N_upper_bound}
        N^\mathbb{C}\leq \left(32\left[\omega^{-1}\left(f,\frac{\varepsilon}{3\sqrt{2m}(1+n/2)}\right)\right]^{-2}+9\right)^{2n},
\end{equation}
where $\omega^{-1}(f,\cdot)$ is the \emph{inverse modulus of continuity} of the function $f$. Since $f$ has Lipschitz constant $K$, we have $\omega^{-1}(f,\delta)={\delta}/{K}$ for $\delta>0$. Substituting this into \eqref{eqn: N_upper_bound} yields~\eqref{lemma1}.
\end{proof}
On the other hand, to approximate the complex-valued mapping $f$ with RVNNs, $f$ needs to be first transformed into a real-valued mapping $\hat{f}:[-1,1]^{2n}\rightarrow \mathbb{R}^{2m}$. Then, a similar result on the network depth can be derived for RVNNs based on~\cite[Proposition 59]{theory2} with its proof omitted for brevity.
\begin{lemma}\label{lemma: rvnn_bound}
    Let $\hat{f}:[-1,1]^{2n}\rightarrow \mathbb{R}^{2m}$ be a $K$-Lipschitz continuous function and $\varepsilon>0$ be given. Then, there exists an RVNN $\Phi^\mathbb{R}$ with width $2n+2m+2$ and depth $N^\mathbb{R}$ such that $\Vert \hat{f}-\Phi^\mathbb{R} \Vert\leq \varepsilon$, where the depth $N^\mathbb{R}$ is of order
    \begin{equation}
        N^\mathbb{R}\sim \mathcal{O}\left(2m\left(\frac{2mK(1+n/2)}{\varepsilon}\right)^{4n}\right).
    \end{equation}
\end{lemma}
Comparing the orders of $N^\mathbb{C}$ and $N^\mathbb{R}$ reveals the following theorem on the depth order of CVNN and its real-valued counterpart.
\begin{theorem}
   \textit{ When the output dimension $m$ is sufficiently large and other parameters are fixed, the required depth of the CVNN, $N^\mathbb{C}$, is smaller than that of the RVNN, $N^\mathbb{R}$, to realize the same approximation precision.}
\end{theorem}
\begin{proof}
    Based on Lemma~\ref{lemma: cvnn_bound}, the order of $N^\mathbb{C}$ can be upper bounded as
    \begin{align}
         N^\mathbb{C}&\sim \mathcal{O}\left(\left(\frac{144mK^2(n+2)^2}{\varepsilon^2}+9\right)^{2n}\right)\nonumber\\
         &< \mathcal{O}\left(\left(\frac{m^2K^2(1+n/2)^2}{\varepsilon^2}\right)^{2n}\right)\nonumber\\
         &=\mathcal{O}\left(\left(\frac{mK(1+n/2)}{\varepsilon}\right)^{4n}\right).
    \end{align}
    Comparing the upper bound with the order $N^\mathbb{R}$ in Lemma~\ref{lemma: rvnn_bound} yields the said relationship.
\end{proof}
Theorem~1 shows that for problems with sufficiently high dimensionality, the CVNN requires fewer layers than the RVNN to approximate a complex-valued mapping, thereby enjoying lower computational complexity without sacrificing task performance.

\section{Complex-Valued Transformers}
\label{sec_cvtrans}
\begin{figure}[t!]
    \centering
    \includegraphics[width=1\linewidth]{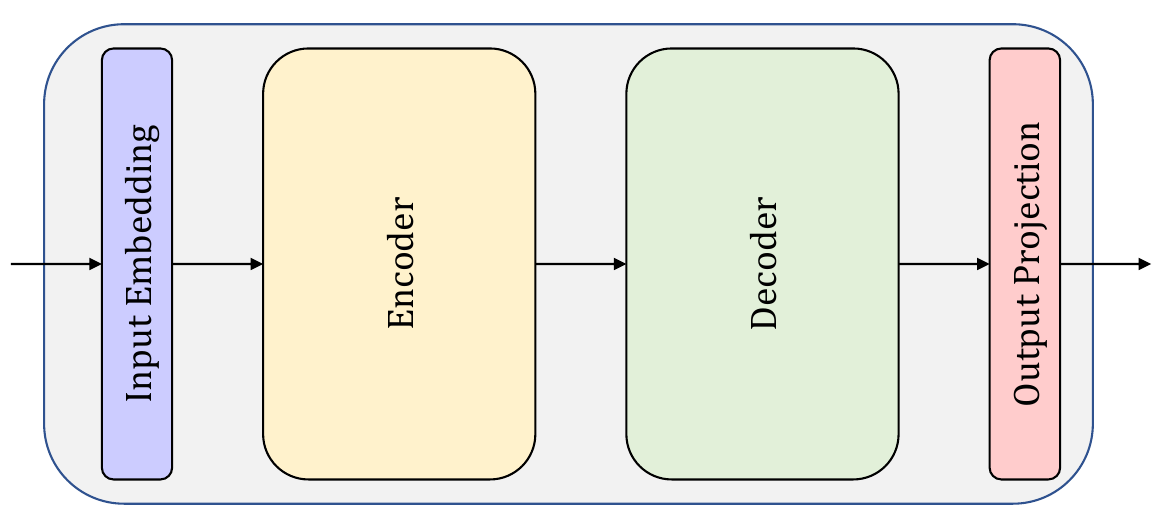}
    \caption{The general architecture of a complex-valued transformer.}
    \label{fig:cvtrans}
\end{figure}
A significant number of modules in wireless communication systems have complex-valued inputs, such as channel state information (CSI) and received signals. Thanks to its native processing of complex-valued signals, the complex-valued transformer holds significant potential for enhancing the efficiency and performance of wireless communication systems. The general architecture of the complex-valued transformer, as shown in Fig.~\ref{fig:cvtrans}, encompasses an input embedding module, a complex-valued transformer that consists of multiple layers of encoders and decoders, and an output projection module. 

The embedding module projects the input signals into an attention space, facilitating subsequent feature extraction. The transformer adopts an encoder-decoder structure, in which the encoder performs feature extraction via the attention mechanism on the input signals, while the decoder decodes and aggregates the encoded features. Subsequently, the output projection module provides a mapping from the complex-valued feature space to the target space. For tasks with real-valued output, such as classification tasks, a $\mathbb{C}2\mathbb{R}$ layer is needed. As the input and output embeddings are based on standard complex-valued neural network modules that are covered in the previous section, we focus on the complex-valued attention below.

% \subsection{Complex-valued multi-head attention mechanism ($\mathbb{C}\operatorname{MHA}$)}
The attention mechanism is a core component of the transformer model, as it enables capturing the relationships between different elements within a matrix. \textcolor{black}{Specifically, in a real-valued transformer, the input to the transformer is a matrix $\mathbf{X}^\mathbb{R} \in \mathbb{R}^{L \times N}$, where $N$ denotes the number of vectors in the sequence,
and $L$ denotes the dimension of each vector. The input matrix is first projected into three distinct spaces of dimension $d$ by three different linear transformations, generating matrices $\mathbf{Q}^\mathbb{R}, \mathbf{K}^\mathbb{R},$ and $\mathbf{V}^\mathbb{R}$. Subsequently, a normalized attention matrix is computed by the scaled dot-product:
\begin{align}
A t t(\mathbf{Q}^\mathbb{R}, \mathbf{K}^\mathbb{R}, \mathbf{V}^\mathbb{R}) \triangleq 
\text{Softmax}({{\mathbf{Q}^\mathbb{R}}^T}\mathbf{K}^\mathbb{R} /{\sqrt{d}}) \mathbf{V}^T,
\label{attn}
\end{align}}to measure the correlation between any two column vectors in $\mathbf{X}^\mathbb{R}$. 

However, in the complex-valued domain, we cannot utilize equation~\eqref{attn} for the attention 
mechanism because the input to $\text{Softmax}(\cdot)$ must be real numbers. With the insight that the operation $\text{Softmax}({\mathbf{Q}^\mathbb{R}}^T {\mathbf{K}^\mathbb{R}}/{\sqrt{d}})$ is to evaluate the similarity between vectors in $\mathbf{Q}^\mathbb{R}$ and $\mathbf{K}^\mathbb{R}$, we modify the term ${\mathbf{Q}^\mathbb{R}}^T \mathbf{K}^\mathbb{R}$ to $\Re({{\mathbf{Q}^{\mathbb{C}}}^H} \mathbf{K}^{\mathbb{C}})$, whose physical meaning is the sum of cosine similarities between corresponding column vectors in $\mathbf{Q}^{\mathbb{C}}$ and $\mathbf{K}^{\mathbb{C}}$, weighted by their respective amplitudes.
Then, the complex-valued attention mechanism is defined as

\begin{small}
\begin{align} 
\mathbb{C} A t t(\mathbf{Q}^{\mathbb{C}}, \mathbf{K}^{\mathbb{C}}, \mathbf{V}^{\mathbb{C}}) \triangleq \text{Softmax}({\Re\left({\mathbf{Q}^{\mathbb{C}}}^H \mathbf{K}^{\mathbb{C}}\right)}/{\sqrt{d}}) {\mathbf{V}^{\mathbb{C}}}^T,
\label{eq:CAtt}
\end{align}
\end{small}where $\mathbf{Q}^{\mathbb{C}}$, $\mathbf{K}^{\mathbb{C}}$, and $\mathbf{V}^{\mathbb{C}}$ are obtained by three different $\mathbb{C}$Linear layers without bias from the complex-valued input $\mathbf{X}^\mathbb{C}$.

For the complex-valued multi-head attention ($\mathbb{C} \operatorname{MHA}$) mechanism, let $T$ denote the number of attention heads. The $t$-th attention head generates query, key and value with their linear transformation weights $\mathbf{W}_t^{\text{Q}}$, $\mathbf{W}_t^{\text{K}}$, and $\mathbf{W}_t^{\text{V}}$, respectively and computes the attention value by~\eqref{eq:CAtt}.
Then, by combining the attention values from $T$ attention heads with an output trainable parameter $\mathbf{W}^{\text{o}} \in \mathbb{C}^{d \times L}$, we obtain the overall $\mathbb{C}\operatorname{MHA}$ computation result:
\begin{small}
\begin{align}
\label{eq:CMHA}
\mathbb{C}\operatorname{MHA}\left(\mathbf{Q}^{\mathbb{C}}, \mathbf{K}^{\mathbb{C}}, \mathbf{V}^{\mathbb{C}}\right) \triangleq  \mathbf{W}^{\text{o}}\text{Concat}(\operatorname{head_1}, ..., \operatorname{head_T}),
\end{align}
\end{small}
where $\operatorname{head_t} = \mathbb{C} A t t (\mathbf{W}_t^{\text{Q}}\mathbf{Q}^{\mathbb{C}}, \mathbf{W}_t^{\text{K}}\mathbf{K}^{\mathbb{C}} , \mathbf{W}_t^{\text{V}}\mathbf{V}^{\mathbb{C}})$.

\section{Application I: Channel Estimation}
\label{case1}
% \subsection{Channel Estimation in OFDM system}
Channel estimation is a critical component in wireless communication systems. It is needed for equalization, beamforming, and interference management.  
In an orthogonal frequency division multiplexing (OFDM) system with $N_\text{f}$ subcarriers, and each data frame spans $N_\text{s}$ OFDM symbols, the data to be transmitted is denoted by $\mathbf{X}\in\mathbb{C}^{N_\text{f}\times N_\text{s}}$, which undergoes inverse fast Fourier transform (IFFT) and cyclic prefix (CP) insertion, and is transmitted over the wireless channel to the receiver.

The receiver then applies CP removal and FFT to obtain the channel-distorted OFDM frame $\mathbf{Y}\in\mathbb{C}^{N_\text{f}\times N_\text{s}}$, which is given by
\begin{equation}
    \mathbf{Y}=\mathbf{H}\circ \mathbf{X}+\mathbf{W},
\end{equation}
where $\mathbf{H}\in\mathbb{C}^{N_\text{f}\times N_\text{s}}$ denotes the channel frequency response at different subcarriers and times, while $\mathbf{W}$ contains independent and identically distributed (i.i.d.) elements following $\mathcal{CN}(0,\sigma^2_W)$ with $\sigma^2_W$ denoting the noise variance, and the operator $\circ$ denotes Hadamard product, i.e., elementwise multiplication. Denoting the pilot symbol positions in $\mathbf{X}$ as $(i,j) \in \Omega $, the least squares (LS) channel estimates at the pilot positions are given by
\begin{equation}
    \hat{\mathbf{H}}_\text{LS}(i,j)=\frac{\mathbf{Y}(i,j)}{\mathbf{X}(i,j)},~ (i,j) \in \Omega.
    \label{HLS}
\end{equation}
\begin{figure}[tb] 
	\centering
	\subfigure[]{ 
		\label{fig_CHA02}
		\includegraphics[width=3.5in,height=1.1in]{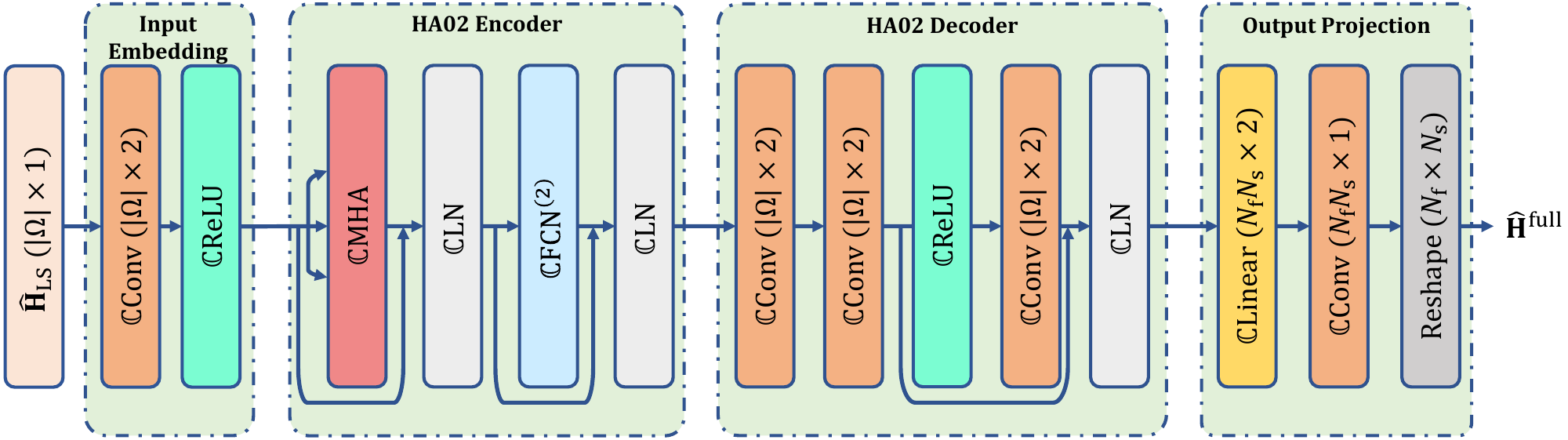}} \hspace{0in}
    \subfigure[]{ 
		\label{fig_ComplexLight}
      \includegraphics[width=3.5in,height=1.1in]{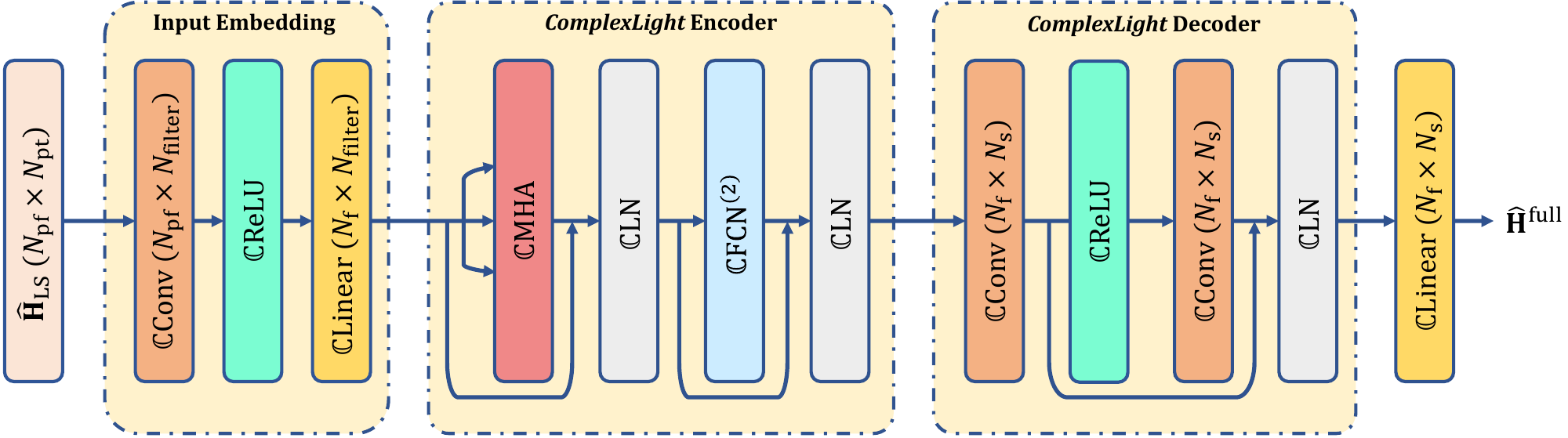}} \hspace{0in}
	\caption{Architectures of (a) $\mathbb{C}$HA02 and (b) the proposed ComplexLight for channel estimation. The output size of each layer keeps that of the preceding layer unless specified in brackets following the layer name.}  \label{fig:cetrans}
\end{figure}
Based on the LS estimates $\hat{\mathbf{H}}_\text{LS}$ at pilot positions, a real-valued attention-based network was proposed for channel estimation at other time-frequency positions, which is called HA02~\cite{luan2023channelformer} and its extension to complex-valued domain is shown in Fig.~\ref{fig_CHA02}. The architecture consists of a standard transformer encoder and a residual convolutional block as the decoder. Notably, the input embedding layer of HA02 flattens the estimates from \eqref{HLS} for subsequent operations, leading to an input dimension $|\Omega| \times 2 \times 1$. On the other hand, the output projection module is a large fully connected layer to change the dimension of decoded representation to $N_\text{f}N_\text{s}$, which costs a huge number of trainable parameters. 

In addition to extending HA02 to the complex-valued domain by replacing various components with their complex-valued counterpart (denoted as Complex HA02), we further propose a lightweight network structure for more efficient parameter use at the output projection layer, and the model is denoted as ComplexLight. We follow~\cite{luan2023channelformer} to use LS estimates $\hat{\mathbf{H}}_\text{LS}$ from \eqref{HLS} as input but without flattening it into one-dimensional structure. This would keep the time-frequency structure of $\hat{\mathbf{H}}_\text{LS}$. Then, time- and frequency-domain transformations can be carried out separately by $\mathbb{C}$Linear and $\mathbb{C}$Conv layers that operate respectively on time and frequency dimensions. Compared to HA02 which employs one large layer operating on a flattened dimension consisting of both time and frequency domain information, the proposed ComplexLight architecture saves a significant amount of parameters.

The overall architecture of the proposed ComplexLight for channel estimation is shown in Fig.~\ref{fig_ComplexLight}, assuming there are $N_{\text{pf}} \times N_{\text{pt}}$ positions in $\Omega$. Initially, $\hat{\mathbf{H}}_\text{LS}$
%, which are widely used as input in channel estimation neural networks, 
are projected into a higher-dimensional space by convolution layer and linear transformation in the complex-valued domain. The $\mathbb{C}$Conv layer utilizes $N_{\text{filter}}$ filters with a kernel size of 3 to project $N_{\text{pt}}$ input channels into higher dimensional space with dimension $N_{\text{filter}}$, while $\mathbb{C}$Linear layer transforms $N_{\text{pf}}$-dimension input to ${N_\text{f}}$-dimension output, resulting in an output $\hat{\mathbf{H}}_\text{em} \in \mathbb{C}^{N_\text{f} \times N_{\text{filter}}}$. 

Subsequently, the resulting output is fed into the encoder that adopts complex-valued attention mechanisms to concentrate on the more significant elements. The encoder consists of three modules: $\mathbb{C} \operatorname{MHA}$ module, $\mathbb{C} \operatorname{LN}$ module and complex-valued feed-forward ($\mathbb{C}\operatorname{FCN}$) module. 

In $\mathbb{C} \operatorname{MHA}$ module, we apply $\mathbb{C}\operatorname{MHA}$ operation with a residual shortcut on $\hat{\mathbf{H}}_\text{em}$:
\begin{align}
\label{eq:H_mha}
\hat{\mathbf{H}}_\text{MHA} = \mathbb{C}\operatorname{MHA}\left(\mathbf{Q}, \mathbf{K}, \mathbf{V}\right)+\hat{\mathbf{H}}_\text{em}.
\end{align}
Then, $\hat{\mathbf{H}}_\text{MHA}$ is passed through $\mathbb{C}$LN operations (described in~\eqref{eq_CLN1}-\eqref{eq_CLN2}), and then sent to $\mathbb{C}\operatorname{FCN}^{(2)}$ module, which first allows for the mapping of features into a higher-dimensional space with dimension $d_f$ to extract deep and nonlinear features. Subsequently, these features are mapped back into a lower-dimensional space $N_\text{f}$, facilitating the overall feature extraction process in the transformer model. Finally, the encoded output with the residual shortcut $\hat{\mathbf{H}}_\text{en}=\mathbb{C}\operatorname{FCN}^{(2)}\left(\hat{\mathbf{H}}_\text{MHA}\right)+\hat{\mathbf{H}}_\text{MHA} \in \mathbb{C}^{{N_\text{f}} \times {N_\text{filter}}}$ of the transformer encoder is normalized by another $\mathbb{C}\operatorname{LN}$.

On the other hand, the decoder module aims to aggregate the encoded feature $\hat{\mathbf{H}}_\text{en}$ to construct the final channel estimation matrix. First, we adjust the number of feature channels from dimension ${N_\text{filter}}$ to ${N_\text{s}}$ through a 1-D $\mathbb{C}$Conv layer with a kernel size of 3, followed by a $\mathbb{C}$ReLU function as shown in the Fig.~\ref{fig_ComplexLight}. We apply another 1-D $\mathbb{C}$Conv layer to further extract the spatial information from the feature representations. This layer has $N_\text{s}$ filters, with each filter corresponding to a kernel size of $3$. Additionally, to avoid the gradient diminishing problem in backpropagation, we include a skip connection in the module. 

Finally, we project the normalized representations $\hat{\mathbf{H}}_\text{de} \in \mathbb{C}^{{N_\text{f}} \times {N_\text{s}}}$, to the space of the full channel matrix $\hat{\mathbf{H}}^\text{full} \in \mathbb{C}^{{N_\text{f}} \times {N_\text{s}}}$ using a $\mathbb{C}$Linear layer:
\begin{align} \label{out_H}
\hat{\mathbf{H}}^\text{full} = \mathbf{W}_{\text{out}}{\hat{\mathbf{H}}_\text{de}} + \mathbf{b}_{\text{out}}, 
\end{align}
where $\mathbf{W}_{\text{out}} \in \mathbb{C}^{N_\text{f} \times N_\text{f}}$ and $\mathbf{b}_{\text{out}} \in \mathbb{C}^{N_\text{s}}$ are the trainable parameters. Note that the number of trainable parameters in the final $\mathbb{C}$Linear layer is significantly reduced compared to HA02.

\textcolor{black}{To train the deep learning models, we employ the loss function}
\begin{equation}
  \mathcal{L}(\hat{\mathbf{H}}^\text{full}, \mathbf{H})\triangleq \frac{1}{N_\text{f} N_\text{s}} \sum_{i=1}^{N_\text{f}} \sum_{j=1}^{N_\text{s}}L\left(\hat{H}^\text{full}_{i j}-H_{i j}\right),
\end{equation}
where $H_{i j}$ is the ground-truth uplink channel at subcarrier $i$ and OFDM symbol $j$ and $\hat{H}^\text{full}_{i j}$ is the corresponding uplink channel predicted by the model. The function $L(\cdot)$ is 
\begin{equation}
  L(x)\triangleq \begin{cases}\frac{1}{2} x^2, & \text { if }|x| \leq 1, \\ |x|-\frac{1}{2}, & \text { otherwise },\end{cases}  
\end{equation}
which is the Huber loss to reduce the impact of outliers. In the training phase, the ground-truth channel $\mathbf{H}$ can be generated via channel models such as the extended typical urban (ETU) model, of which the details will be provided in Section~\ref{sec: exp_setting}. We follow a standard mini-batch training with a stochastic gradient descent algorithm and Adam optimizer to minimize the Huber loss function during the training phase.

\section{Application II: User Activity Detection}
\label{case2}
In grant-free access, there are a lot of potential users but only a small fraction of them are active in a given time slot. With each active user sending a pre-assigned pilot sequence before the actual data, the base station aims to detect which users are active based on the received complex-valued signal. This makes activity detection in grant-free access essentially a classification task.  

Consider a single-cell scenario, where the BS has $M$ antennas and serves $N$ users, each with a single antenna. 
With a unique length$-L$ signature sequence ${\mathbf{s}}_{n} \in \mathbb{C}^{L}$ 
% =\left[ s_{n, 1}, s_{n, 2}, \ldots, s_{n, L}\right]^{T} 
assigned to user~$n$ and the $n$th user to the BS small-scale fading channel denoted by $h_n$, the received signal at the BS is given by
\begin{align} \label{received_signal}
\mathbf{Y}&=\sum_{n=1}^N \sqrt{\eta_n g_n} a_n  \mathbf{s}_n \mathbf{h}_n^T+\mathbf{W},
\end{align}
where $g_{n}$ is the large-scale fading coefficient of the $n$-th user to the BS, $a_n$ is the activity indicator taking 0 (if the user is inactive) or 1 (if the user is active), and $\eta_{n}$ is the transmit power of the $n$-th user. The elements of $\mathbf{W} \in \mathbb{C}^{L\times M}$ are i.i.d. Gaussian noise at the BS.

\textcolor{black}{It is observed that detecting user activities (which users' signals on the right-hand side of~\eqref{received_signal} are non-zero) is equivalent to the BS deciphering which signature sequences are contained in $\mathbf{Y}$. Therefore, it is natural that we use $\mathbf{B} \triangleq \left[\sqrt{\eta_1 g_1}  \mathbf{s}_1, \right.$ $\left.
\sqrt{\eta_2 g_2}  \mathbf{s}_2, \dots, \sqrt{\eta_N g_N}  \mathbf{s}_N\right]$ and $\mathbf{Y}$ as input features and the neural network is expected to learn the relevance between different columns of $\mathbf{B}$ and $\mathbf{Y}$. Following the principle of the real-valued transformer for the same application~\cite{LY_transformer}, we propose a complex-valued transformer for activity detection. } 

The overall architecture is illustrated in Fig.~\ref{fig:AC_ctrans}, which comprises an embedding layer, $Z$ layers of encoders, an aggregation decoder, and a $\mathbb{C}2\mathbb{R}$ layer. The embedding layer, encoders, and decoder are obtained by replacing the real-valued network from~\cite{LY_transformer} with a complex-valued version, while the $\mathbb{C}2\mathbb{R}$ layer is for converting the final complex-valued features into a probability vector.
\begin{figure}
    \centering
    \includegraphics[width=1\linewidth]{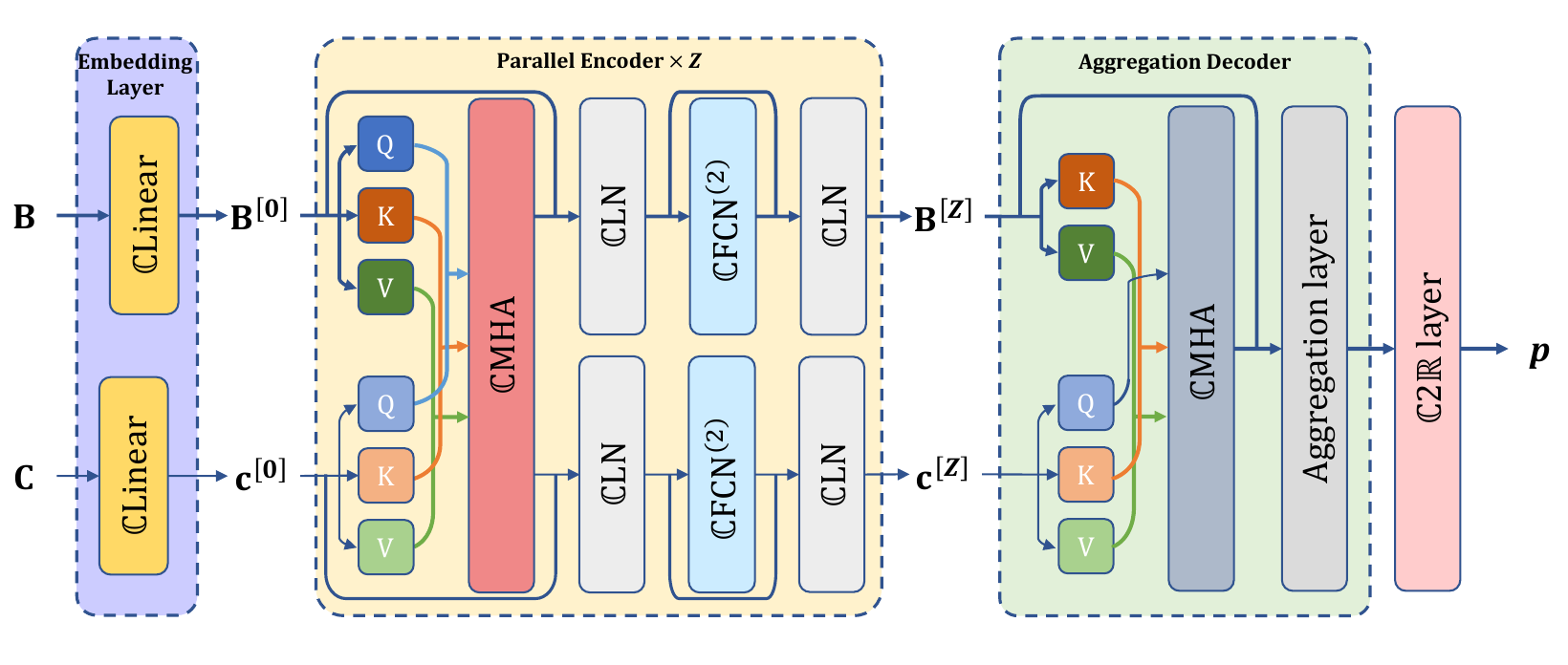}
    \caption{The overall architecture of the complex-valued heterogeneous transformer in activity detection.}
    \label{fig:AC_ctrans}
\end{figure}

To make the model scale adaptable to the number of BS antennas $M$, the sample covariance matrix $\mathbf{C} = \mathbf{Y}\mathbf{Y}^H$ instead of $\mathbf{Y}$, together with the signature matrix $\mathbf{B} $, are input to the embedding layer, which applies two $\mathbb{C}\operatorname{Linear}$ to project $\mathbf{B}$ and $\mathbf{C}$ into $\mathbf{B}^{[0]}$ and $\mathbf{c}^{[0]}$ respectively. 

Next, the transformer encoder consists of $Z$ layers, with the the $z$-th encoder layer's inputs denoted as $\mathbf{B}^{[z-1]}$ and $\mathbf{c}^{[z-1]}$. In the encoder structure, $\mathbf{B}^{[z-1]}$ and $\mathbf{c}^{[z-1]}$ are first respectively projected into queries  ($\mathbf{Q}_\text{b}$ and $\mathbf{q}_\text{c}$), keys ($\mathbf{K}_\text{b}$ and $\mathbf{k}_\text{c}$) and values ($\mathbf{V}_\text{b}$ and $\mathbf{v}_\text{c}$) by six $\mathbb{C}$Linear transformations. 

Subsequently, we obtain the $\mathbb{C}\operatorname{MHA}$ computation result 
$\mathbb{C}\operatorname{MHA}([\mathbf{Q}_\text{b}\; \mathbf{q}_\text{c}], [\mathbf{K}_\text{b}\; \mathbf{k}_\text{c}], [\mathbf{V}_\text{b}\; \mathbf{v}_\text{c}])$. The output of this CHMA module is split into the first $N$ columns and $(N+1)$-th column, summed with $\mathbf{B}^{[z-1]}$ and $\mathbf{c}^{[z-1]}$ respectively, and fed into $\mathbb{C}$LN layers. The two parallel outputs are further operated by residual $\mathbb{C}\operatorname{FCN}^{(2)}$ modules and $\mathbb{C}$LN layers to generate the final encoder output $\mathbf{B}^{[z]}$ and $\mathbf{c}^{[z]}$. It is emphasized that the transformer in this application is a heterogeneous transformer and more complicated than the standard transformer. The information from $\mathbf{B}$ and $\mathbf{C}$ are mixed while having their individual representations in each layer.

After the $Z$ encoder layers, the feature representations $\mathbf{B}^{[Z]}$ and $\mathbf{c}^{[Z]}$ are passed to the aggregation decoder to generate the final result. Here, the query $\tilde{\mathbf{q}}_\text{c}$ is solely projected from $\mathbf{c}^{[Z]}$, while the keys ($\tilde{\mathbf{K}}_\text{b}$ and $\tilde{\mathbf{k}}_\text{c}$) and values ($\tilde{\mathbf{V}}_\text{b}$ and $\tilde{\mathbf{v}}_\text{c}$) are projected from both $\mathbf{B}^{[Z]}$ and $\mathbf{c}^{[Z]}$. Then, the context vector $\mathbf{x}_d$, which reflects which user signature sequence to pay attention to based on the received signal, is given by
\begin{align}
\mathbf{x}_{\text{d}} = \mathbb{C}\operatorname{MHA}(\tilde{\mathbf{q}}_\text{c}, [\tilde{\mathbf{K}}_\text{b}\; \tilde{\mathbf{k}}_\text{c}], [\tilde{\mathbf{V}}_\text{b}\; \tilde{\mathbf{v}}_\text{c}]).
\label{xd}
\end{align}
Finally, the $\mathbb{C}2\mathbb{R}$ layer defined in~\eqref{eq_C2R} provides a connection from the complex-valued features~\eqref{xd} to the active probability of user~$n$ before applying the loss function.

During the training phase, we generate $J$ training samples as detailed in Section~\ref{ac_results}. Each sample is represented as a triplet $(\mathbf{C}^{(j)}, \mathbf{B}^{(j)}, {a_n^{(j)}}_{n=1}^{N})$ for $j = 1,\ldots,J$, where $a_{n}^{(j)}$ is the ground-truth. We follow a standard mini-batch training with a stochastic gradient descent algorithm to minimize the weighted cross entropy-based loss function as in~\cite{LY_transformer}.

\section{Application III: Joint Design of Pilot Sequence, Feedback Quantization, and Precoder}
\label{case3}
Consider a downlink FDD massive MIMO system consisting of one BS with $N$ transmit antennas and $K$ single-antenna users with $K$~\textless~$N$. 
The BS utilizes the uplink channel $\mathbf{H}^{\mathrm{UL}} \triangleq [\mathbf{h}_1^{\mathrm{UL}}, \mathbf{h}_2^{\mathrm{UL}}, \cdots, \mathbf{h}_K^{\mathrm{UL}}] \in \mathbb{C}^{N \times K}$ to design the downlink pilot signal $\mathbf{P} \triangleq \left[\mathbf{p}_1, \mathbf{p}_2, \cdots, \mathbf{p}_M\right] \in \mathbb{C}^{N \times M}$ of length $M$, which can be expressed as  
\begin{align}
 \mathbf{P}=\mathcal{D}(\mathbf{H}^{\mathrm{UL}}),    
\end{align}
where the function $\mathcal{D}(\cdot): \mathbb{C}^{N \times K} \rightarrow\mathbb{C}^{N \times M}$ indicates the mapping function from the uplink channel to the downlink pilot design. 

Then, the BS transmits the pilot signal $\mathbf{P}$ to all users, adhering to the total power constraint among all BS antennas with $\left\|\mathbf{p}_m\right\|_2^2 \leqslant \rho$, where $\rho$ represents the maximum transmit power. The received signal $\mathbf{y}_k \in \mathbb{C}^{M}$ at the $k$th user is then given by
\begin{align} 
\label{y_k}
\mathbf{y}_k=\mathbf{P}^H \mathbf{h}_k^{\mathrm{DL}}+\mathbf{n}_k, k=1,2, \cdots, K,
\end{align} 
where $\mathbf{h}_k^{\mathrm{DL}} \in \mathbb{C}^{N \times 1}$ denotes the downlink channel between the $k$th user and BS, and $\mathbf{n}_k \thicksim \mathcal{CN}\left(0, \sigma^2\right)$ is the complex-valued additive white gaussian noise ($\mathrm{AWGN}$). 

Once $\mathbf{y}_k$ is received, each user quantizes the received signal $\mathbf{y}_k$ into information bits to describe the downlink channel. This quantization process is expressed as:
\begin{align}
 \mathbf{q}_k=\mathcal{F}_k\left(\mathbf{y}_k\right), k=1,2, \cdots, K,
\end{align} 
where the function $\mathcal{F}_k(\cdot): \mathbb{C}^{M} \rightarrow\{-1,+1\}^{B}$ represents the quantization scheme adopted at user $k$. The quantized information bits $\mathbf{q}_k$ are then transmitted to the BS through the feedback link, where we assume that the feedback is perfect without any error and significant overhead. 

Subsequently, the BS designs the multiuser downlink precoding matrix $\mathbf{V} \triangleq 
\left[\mathbf{v}_1, \mathbf{v}_2, 
\cdots, \mathbf{v}_K\right]$ from the quantized information bits of multiple users $\mathbf{Q} \triangleq \left[\mathbf{q}_1, \mathbf{q}_2, \cdots, \mathbf{q}_K\right]$ and the uplink channel $\mathbf{H}^{\mathrm{UL}}$,
\begin{align}
 \mathbf{V}=\mathcal{P}(\mathbf{Q}, \mathbf{H}^{\mathrm{UL}}),    
\end{align}
where the function $\mathcal{P}(\cdot):\{-1,+1\}^{B \times K} \times \mathbb{C}^{N \times K} \rightarrow \mathbb{C}^{N \times K}$ denotes the precoding design scheme. Furthermore, the precoding matrix $\mathbf{V}$ must satisfy the transmit power constraint  $\operatorname{Tr}\left(\mathbf{V}^H \mathbf{V}\right) \leqslant \rho$.

Aiming at maximizing the downlink sum-rate of all users, the uplink-assisted joint pilot, feedback quantization and downlink precoding design problem is described as:
\begin{equation}
\begin{array}{lll}
&\mathop{\text{max}}\limits_{\mathcal{D}(\cdot),\{\mathcal{F}_k(\cdot)\}, \mathcal{P}(\cdot)} & \sum_{k=1}^K \log _2\left(1+\frac{\left|\mathbf{v}_k^H \mathbf{h}_k^{\mathrm{DL}}\right|^2}{\sum_{k^{\prime} \neq k}\left|\mathbf{v}_{k^{\prime}}^H \mathbf{h}_k^{\mathrm{DL}}\right|^2+\sigma^2}\right), \\  \\
&\quad\quad\,\,\text {s.t.} & \mathbf{P}=\mathcal{D}\left(\mathbf{H}^{\mathrm{UL}}\right), \\
&\quad& \mathbf{q}_k=\mathcal{F}_k\left(\mathbf{y}_k\right), \quad k=1,2, \cdots, K, \\
&\quad& \mathbf{V}=\mathcal{P}\left(\mathbf{Q}, \mathbf{H}^{\mathrm{UL}}\right), \\
&\quad& \operatorname{Tr}\left(\mathbf{V}^H \mathbf{V}\right) \leqslant \rho .
\end{array}
\label{precoding_loss}
\end{equation}
A real-valued transformer-based scheme~\cite{jiang2023transformer} has been proposed to solve \eqref{precoding_loss}. We extend it to the complex-valued domain, employ CVNNs to construct the mapping functions $\mathcal{D}(\cdot),\mathcal{F}_k(\cdot)$, and $\mathcal{P}(\cdot)$ to achieve end-to-end joint optimization in an unsupervised way. 

The overall scheme is depicted in Fig. \ref{fig:precoding_scheme}, which encompasses modules for pilot design, feedback quantization, and precoding design. To simplify the problem without sacrificing generality, we assume that each user deploys the same CVNN, i.e., $\mathcal{F}_k(\cdot) = \mathcal{F}(\cdot), \forall k=1, 2, \cdots, K$. Next, we describe how we use the CVNNs model to realize each module.
\begin{figure}[t!]
    \centering
    \includegraphics[width=1\linewidth]{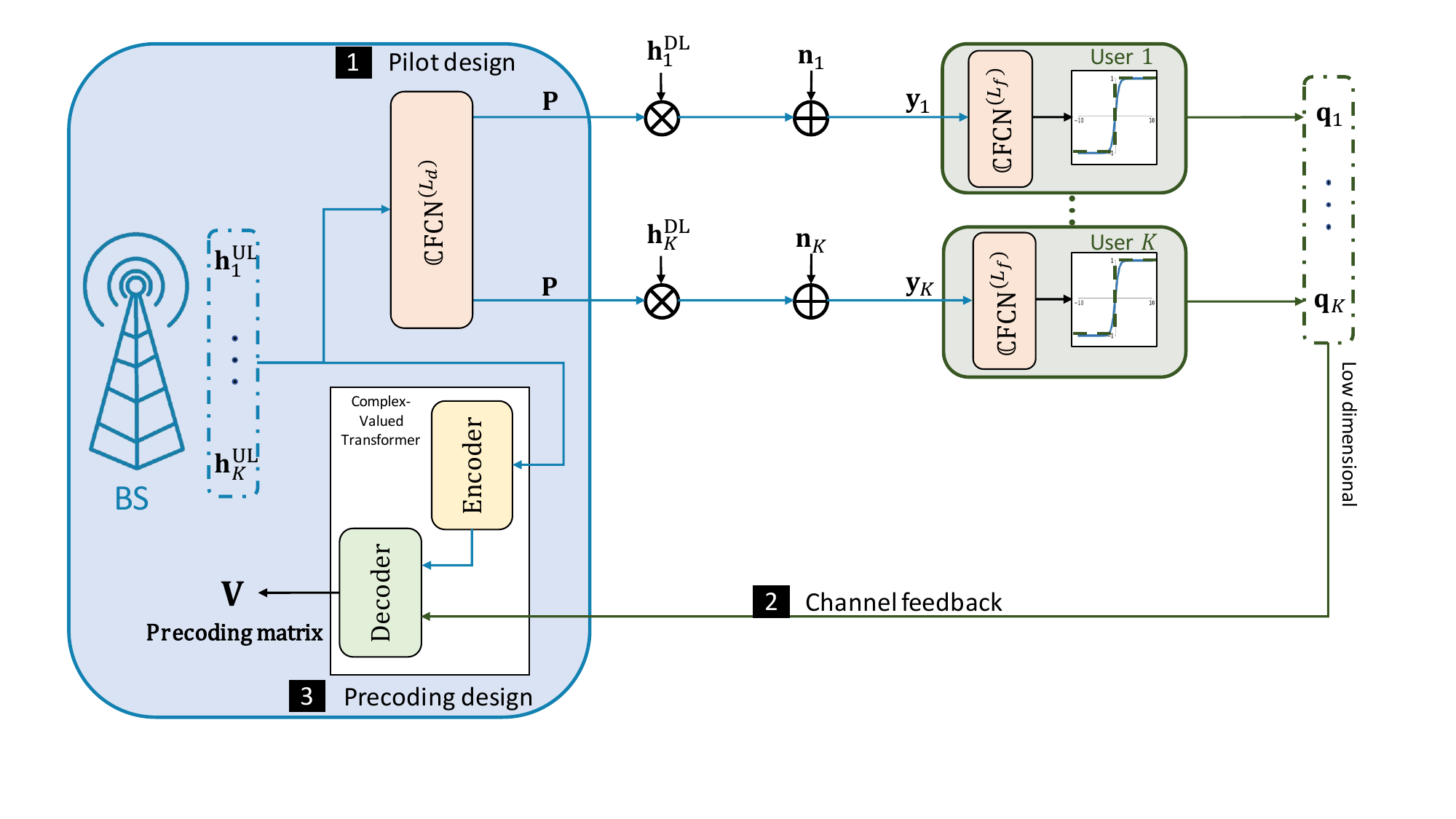}
    \caption{Uplink-aided CVNNs-based joint pilot sequence, feedback quantization and downlink precoder.}
    \label{fig:precoding_scheme}
\end{figure}

\subsubsection{Pilot Design}
\textcolor{black}{The downlink pilot design module is realized by a $L_d$-layer complex-valued fully connected neural network ($\mathbb{C}\operatorname{FCN}^{(L_d)}$), followed by a power normalization function to satisfy the transmit power constraints:
\begin{align}
    \mathbf{P} = \delta_d \left( \mathbb{C}\operatorname{FCN}^{(L_d)}(\mathbf{H}^{\mathrm{UL}})\right),
\end{align}}where $\delta_d$ is the normalization such that the norm square of every column of $\mathbf{P}$ is smaller than or equal to $\rho$.
\subsubsection{Channel Feedback}
\textcolor{black}{After receiving the pilot signal through~\eqref{y_k}, each user utilizes another $\mathbb{C}\operatorname{FCN}^{(L_f)}$ with $L_f$ layers to extract the downlink channel information from the received pilot signal $\mathbf{y}_k$. The user then quantifies the output into binary outcomes. The overall process can be expressed as
\begin{align}
\mathbf{q}_k=\delta_f \left( \mathbb{C}\operatorname{FCN}^{(L_f)}\left(\mathbf{y}_k\right)\right).
\end{align}}Since we need binary outcomes in $\mathbf{q}_k$, we propose to apply the tanh function on the real part of the input in the training phase:
\begin{equation}
    \delta_f(\mathbf{x})=\operatorname{tanh}(\tau\Re(\mathbf{x})),
\end{equation}
\textcolor{black}{Notice that $\operatorname{tanh}(\tau x)$ is a differentiable approximation of $\operatorname{sgn}(x)$, with $\lim_{\tau\to \infty} \operatorname{tanh}(\tau x)=\operatorname{sgn}(x)$ for all $x\neq 0$. After training, we still use the sign function for inference. In experiments, we set $\tau=10$, as a too-large value might slow down the training process due to small gradients. }
\subsubsection{Precoder Design}
Next, the BS designs the downlink precoder based on the uplink channel and feedback bits from each user. For this module, we develop a complex-valued version of the transformer in~\cite{jiang2023transformer}. The uplink CSI and the feedback bits are treated as multimodal data relating to the same downlink CSI and these two types of data are merged for the downlink precoding design. This complex-valued transformer consists of an input embedding module, $L_p$ layers of encoders, $M_p$ layers of decoders, and a $\mathbb{C}$Linear layer for final output projection, as depicted in Fig.\ref{fig:precoding_model}. 
\begin{figure}[t!]
    \centering
    \includegraphics[width=1\linewidth]{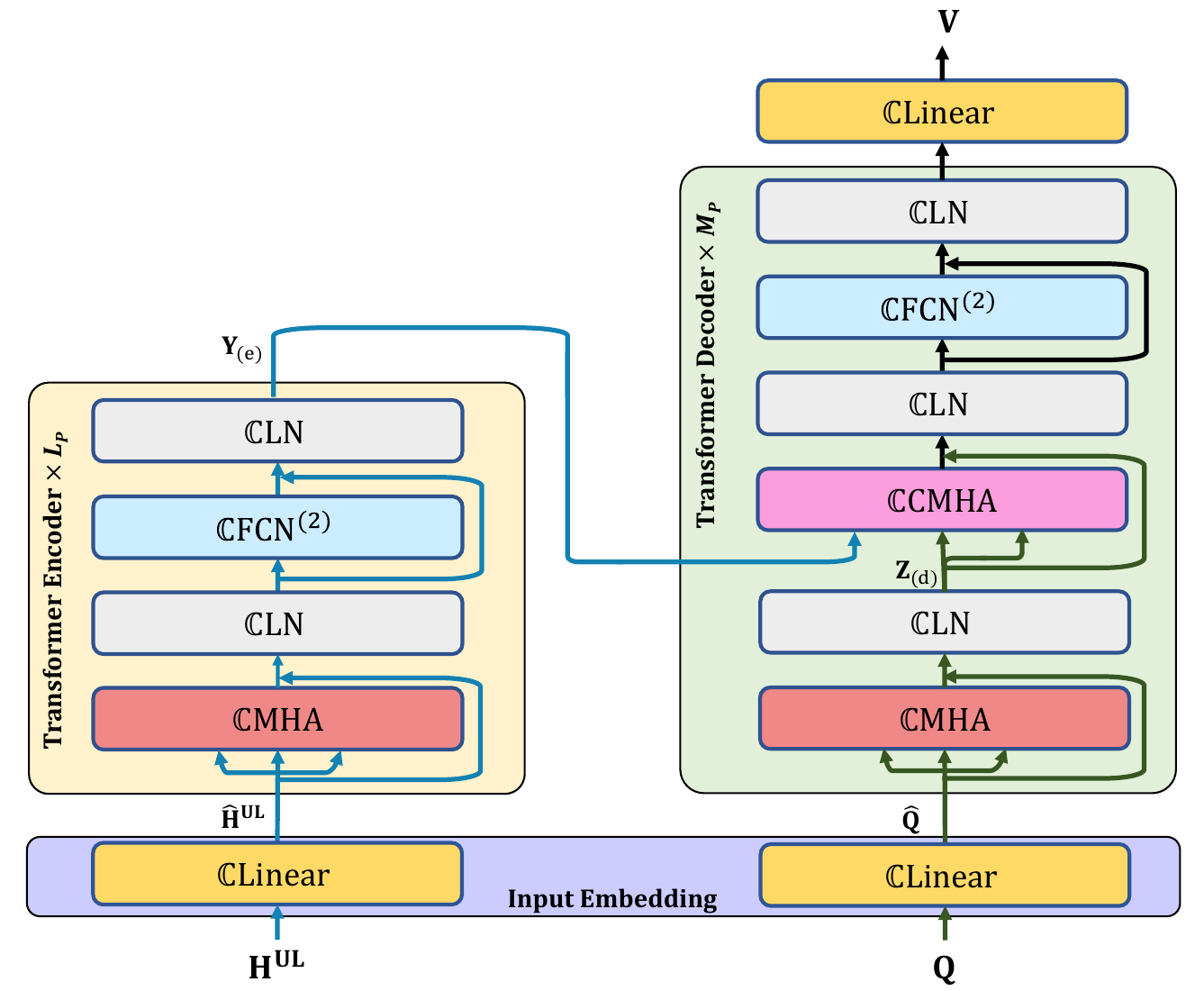}
    \caption{The architecture of complex-valued transformer for downlink precoding design.}
    \label{fig:precoding_model}
\end{figure}

Firstly, the uplink channel $\mathbf{H}^{\mathrm{UL}} \in \mathbb{C}^{N \times K}$ is projected into the embedding dimension $d_p$ by $\mathbb{C}$Linear with the output denoted by $\hat{\mathbf{H}}^{\mathrm{UL}} \in \mathbb{C}^{{d_p} \times K}$. The encoder then computes the complex-valued multi-head attention using~\eqref{eq:CMHA}. To address the issue of vanishing gradient, skip connections and $\mathbb{C}\operatorname{LN}$ are implemented. This process is written as 
\begin{align}
\mathbf{Z}_{(e)}=\mathbb{C}\operatorname{LN}\left(\hat{\mathbf{H}}^{\mathrm{UL}}+ \mathbb{C}\operatorname{MHA}(\hat{\mathbf{H}}^{\mathrm{UL}})\right).
\end{align}
Furthermore, two layers of $\mathbb{C}\operatorname{FCN}$ are employed to further extract and compress the features of each user, which is represented as
\begin{align}
\mathbf{Y}_{(e)}=\mathbb{C}\operatorname{LN}\left(\mathbf{Z}_{(e)}+\mathbb{C}\operatorname{FCN}^{(2)}\left(\mathbf{Z}_{(e)}\right)\right).
\end{align}
Similarly, the decoder firstly employs a linear transformation on the feedback information bits $\mathbf{Q} \in \mathbb{R}^{B \times K}$ and generate $\hat{\mathbf{Q}} \in \mathbb{C}^{{d_p} \times K}$. Then, it is processed by $\mathbb{C}\operatorname{MHA}$ as well as $\mathbb{C}\operatorname{LN}$ operation:
\begin{align}
\mathbf{Z}_{(d)}=\mathbb{C}\operatorname{LN}\left(\hat{\mathbf{Q}}+ \mathbb{C}\operatorname{MHA}(\hat{\mathbf{Q}})\right).
\end{align}

Subsequently, a complex-valued cross-attention mechanism ($\mathbb{C}\text{CMHA}$) is established between the encoder and decoder, utilizing the encoded channel feature $\mathbf{Y}_{(e)}$ and the extracted information bit feature $\mathbf{Z}_{(d)}$ from the decoder. Specifically, in the cross-attention mechanism,  the key matrix is derived from the encoded feature $\mathbf{Y}_{(e)}$, while the query matrix and value matrix are obtained from the information bits feature $\mathbf{Z}_{(d)}$ extracted by the decoder. The process can be described by
% \begin{small}
\begin{equation}
    \mathbf{Z}_{(v)}=\mathbb{C}\operatorname{LN}(\mathbf{Z}_{(d)}+ \mathbb{C}\operatorname{MHA}(\mathbf{Q}_{\mathbf{Z}_{(d)}}, \mathbf{K}_{\mathbf{Y}_{(e)}}, \mathbf{V}_{\mathbf{Z}_{(d)}}),
\end{equation}
% \end{small}
where, $ \mathbf{Q}_{\mathbf{Z}_{(d)}}$ and $ \mathbf{V}_{\mathbf{Z}_{(d)}}$ are the query matrix and the value matrix generated by two linear transformations from $\mathbf{Z}_{(d)}$, while $\mathbf{K}_{\mathbf{Y}_{(e)}}$ is the key matrix obtained from another linear transformation from $\mathbf{Y}_{(e)}$. The extracted fusion features are then subjected to skip connection and $\mathbb{C}\operatorname{LN}$ operations to obtain the final precoded matrix
\begin{align}
\mathbf{V}= \delta_v \left( \mathbb{C}\text{Linear}(\mathbb{C}\operatorname{LN}(\mathbf{Z}_{(v)}+\mathbb{C}\operatorname{FCN}^{(2)}(\mathbf{Z}_{(v)}))) \right),
\end{align}
where $\delta_v$ is the normalization ensuring that $\mathbf{V}=[\mathbf{v}_1, \cdots, \mathbf{v}_K] 
\in \mathbb{C}^{N \times K}$ satisfies the BS transmit power constraint $\operatorname{Tr}\left(\mathbf{V}^H \mathbf{V}\right) \leqslant \rho $, and $\mathbb{C}\text{Linear}$ is used for output projection from attention space with dimension $d_p$ to precoder space with dimension $N$. 

To train the complex-valued neural network in an end-to-end fashion, the loss function is set as the negative objective function of~\eqref{precoding_loss}. We follow a standard mini-batch training with a stochastic gradient descent algorithm to minimize the loss function. 
\section{Experimental Results}
\begin{table*}[t]
\caption{List of Compared Models and Their Properties.}
\label{tab:channel_estimation}
\resizebox{\textwidth}{!}{
\begin{tabular}{ccccc}
\hline
Channel Estimation & Input Dimension                               & Output Dimension    & Real-valued Parameters \\ \hline
HA02~\cite{luan2023channelformer}                        & $({N_{\text{pt}}{N_\text{f}}}/{2}) \times 2 \times 1$         & ${N_\text{s}{N_\text{f}}} \times 2$             & 105,607    \\
ComplexHA02               & $({N_{\text{pt}}{N_\text{f}}}/{2}) \times 1$         & ${N_\text{s}{N_\text{f}}} \times 1$              & 209,778    \\
RealLight                  & $({{N_\text{f}}}/{2}) \times N_{\text {pilot}}\times 2$ & $ N_\text{s} \times N_\text{f} \times 2$          & 41,288     \\
ComplexLight               & $({{N_\text{f}}}/{2}) \times N_{\text {pilot}}$         & $ N_\text{s} \times N_\text{f}$                   & 80,076     \\ 
RealLight (comparable)               & $({{N_\text{f}}}/{2}) \times N_{\text {pilot}}\times 2$ & $ N_\text{s} \times N_\text{f} \times 2$        & 82,904     \\ \hline
% Activity Detection\\ \hline 
User's Activity Detection             & Input                                   & Output Dimension  & Real-valued Parameters \\ \hline
\multirow{2}{*}{Real-valued transformer~\cite{LY_transformer}} &
  \multirow{2}{*}{\begin{tabular}[c]{@{}c@{}}~~$\mathbf{B} \in \mathbb{R}^{N \times 2L}$\\ ~~$\mathbf{C}\in \mathbb{R}^{1\times2L^2}$\end{tabular}} &
    \multirow{8}{*}{$N \times 1$} &

  \multirow{2}{*}{370,176} \\ 
  \\
\multirow{2}{*}{Complex-valued transformer} &
 \multirow{2}{*}{\begin{tabular}[c]{@{}c@{}}~$\mathbf{B} \in \mathbb{C}^{N \times L}$\\ ~$\mathbf{C}\in \mathbb{C}^{1\times L^2}$\end{tabular}} &
&

  \multirow{2}{*}{567,555} \\
    \\

\multirow{2}{*}{Real-valued transformer (comparable)} &
  \multirow{2}{*}{\begin{tabular}[c]{@{}c@{}}~~$\mathbf{B} \in \mathbb{R}^{N \times 2L}$\\ ~~$\mathbf{C}\in \mathbb{R}^{1\times2L^2}$\end{tabular}} &
   &
   
  \multirow{2}{*}{565,596} \\  
  \\
                        % &$\mathbf{C}\in \mathbb{C}^{1 \times L^2}$ &        &                        \\
 \multirow{2}{*}{Complex-valued MLP} &
  \multirow{2}{*}{\begin{tabular}[c]{@{}c@{}}~$\mathbf{B} \in \mathbb{C}^{N \times L}$\\ ~$\mathbf{C}\in \mathbb{C}^{1\times L^2}$\end{tabular}} &
   &
   
  \multirow{2}{*}{6,737,355} \\ 

  \\ \hline
  % Precoding Design\\ \hline
Joint Design of Pilot, Feedback Quantization, and Precoder                 & Input Dimension   & Output Dimension       & Real-valued Parameters \\ \hline
Real-valued transformer~\cite{jiang2023transformer}     & $K \times 2N$ & $K \times 2N$  & 12,356,536 \\
Complex-valued transformer & $K \times N$  & $K \times N$   & 23,087,472 \\
Real-valued transformer (comparable) & $K \times 2N$ & $K \times 2N$     & 23,974,328               \\
DNN~\cite{sohrabi2021deep}                        & $K \times 2N$ & $K \times 2N$  & 5,850,040  \\ \hline
\end{tabular}}
\label{tab_benchmarks}
\end{table*}
The details of the deep-learning benchmarks in different application examples are summarized in Table~\ref{tab_benchmarks}. In each application, we use the real-valued counterpart as a baseline. It is noteworthy that complex-valued models often have a larger number of parameters than RVNNs. Therefore, we also include a real-valued model with a comparable number of parameters to our complex-valued transformer-based method (denoted as *(comparable)). Other benchmarks will be introduced in each respective case. We evaluate the various methods in terms of inference performance, generalization ability, and performance under decreased training data for three applications.

\label{Sec_Results}
\subsection{Performance on Channel Estimation}\label{sec: exp_setting}
\begin{figure}[t!]
    \centering
    \includegraphics[width=0.45\textwidth]{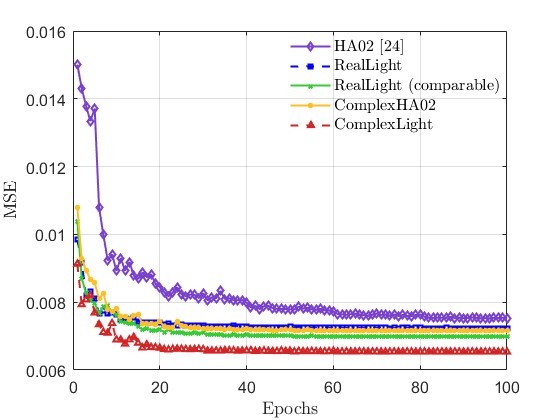}
    \caption{Channel estimation validation loss in the training stage.}
    \label{figCE_loss}
    \end{figure}
    
For the training phase, 125,000 samples are generated with the ETU channel model with signal-to-noise ratio (SNR) from 5 dB to 25 dB and maximum Doppler shift from 0 Hz to 97 Hz, of which 95\% are used for training, and the remaining 5\% are validation set. 

The validation losses versus number of epoch during training is shown in Fig.~\ref{figCE_loss}. It can be seen that the loss curves of all methods decrease quickly and converge gradually as the training epoch increases. Upon comparing various methods, it is evident that the proposed complex-valued transformer frameworks, ComplexHA02 and ComplexLight, surpass the RVNN counterparts both in faster convergence and lower final converged validation losses. Moreover, our proposed ComplexLight demonstrates the swiftest convergence and maintains the smallest loss consistently.
\begin{figure*}[ht] 
	\centering
	\subfigure[]{ 
		\label{figCE_SNR}
		\includegraphics[width=0.3\textwidth]{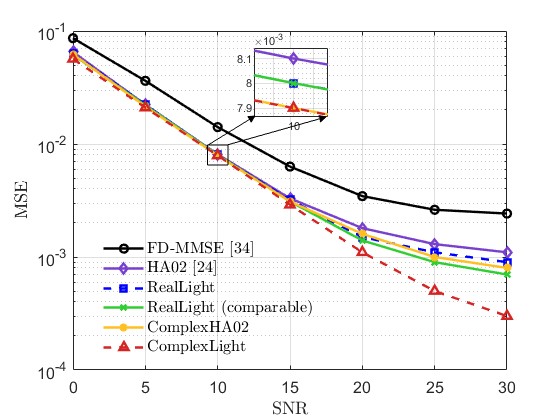}} \hspace{0in}
    \subfigure[]{ 
		\label{figCE_Dop}
      \includegraphics[width=0.3\textwidth]{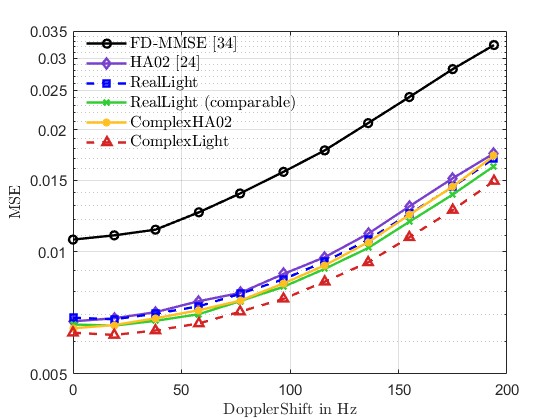}} \hspace{0in}
      \subfigure[]{ 
		\label{figreducesample}
      \includegraphics[width=0.3\textwidth]{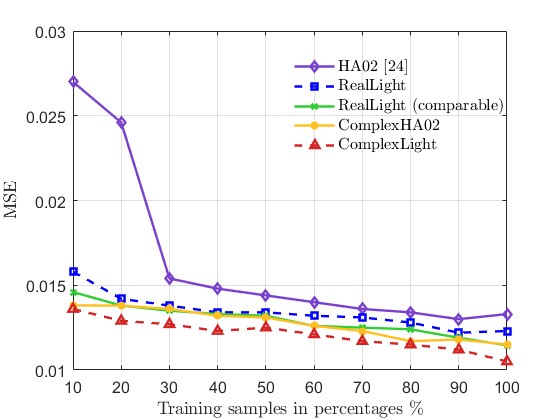}} \hspace{0in}
	\caption{The experimental results of channel estimation: (a) MSE versus SNR, (b) MSE versus doppler shift and (c) MSE on the test set when training on limited data.}  \label{figCE}
\end{figure*}

\textcolor{black}{The performance of various trained models is evaluated on two types of test sets generated with the ETU model. The first one extends the uplink SNR to a larger range than the training set, i.e., $0$--$30$ dB, while the second one expands the range of the Doppler shift to $0$--$194$ Hz. Such mismatches between training and testing data are common in real-world scenarios and demand generalization abilities from the trained models. In two test sets, we generate 5000 samples for each SNR level. The performance on test sets is measured by estimation mean square error (MSE), which is defined as:}
\begin{equation}
\operatorname{MSE}(\hat{\mathbf{H}}^\text{full}, \mathbf{H})=\frac{1}{N_\text{f} N_\text{s}} \sum_{i=1}^{N_\text{f}} \sum_{j=1}^{N_\text{s}}\left|\hat{H}^\text{full}_{i j}-H_{i j}\right|^2.
\end{equation}

The experimental results of channel estimation are depicted in Fig.~\ref{figCE}.
From Fig.~\ref{figCE_SNR}, it is evident that the proposed ComplexLight demonstrates superior performance in all SNRs, with the performance gap increasing further as SNR increases. Notably, all the deep learning-based methods outperform the FD-MMSE method \cite{dowler2003FD-MMSE}, with our ComplexLight achieving the lowest MSE. This is primarily because neural network-based algorithms learn the underlying pattern, including time- and frequency-domain correlation, of OFDM channels through the training on the complete ground-truth channel matrix, while FD-MMSE~\cite{dowler2003FD-MMSE} is a linear estimator that only utilizes the information from the received pilot symbols.

On the other hand, as depicted in Fig.~\ref{figCE_Dop}, our ComplexLight exhibits excellent generalizability compared to other approaches for the whole Doppler shift range. Similarly, all the deep learning-based methods outperform the FD-MMSE method \cite{dowler2003FD-MMSE}, and ComplexLight maintains the lowest MSE.

\textcolor{black}{Finally, we compare the performance of various models when the training data is limited. We vary the percentage of training samples used in model training and evaluate the MSE on the test sets averaged over the SNR levels and Doppler shifts. The result is plotted in Fig.~\ref{figreducesample}. One can observe that CVNNs suffer less impact from the reduction in training samples and consistently require less training data to reach a certain performance compared to their real-valued counterparts. In particular, training ComplexLight using 20\% of the training data leads to the same MSE as that of the RealLight (comparable) trained with 55\% of the training data. This shows an obvious advantage of using complex-valued neural networks. }

\subsection{Performance on User Activity Detection}
\label{ac_results}
In this case, both training and test samples are generated as follows \textcolor{black}{unless otherwise specified}: the number of IoT users $N$ is set as $100$, with the active probability being $0.1$. The number of BS-antennas $M$ is set as $32$. The large-scale coefficient is modeled as $128.1 + 37.6\log_{10}(d)$ in dB, where $d$ is the BS-user distance in kilometers. Each signature symbol $s_{n,l}$ is sampled from $\mathcal{C} \mathcal{N}(0,1)$ and the pilot length $L$ is $8$. The background noise has a power spectral density of -169 dBm/Hz across 10 MHz, and the transmit power $\eta_n$ is set as 23 dBm. The received signal follows~\eqref{received_signal}, where each user's activity is a realization of Bernoulli random variable. Besides, the hyperparameters of our complex-valued model are given in Table.~\ref{trans_p}. 

\begin{table*}[] 
\centering 
\caption{Hyper-Parameters of the Complex-valued Transformer in User Activity Detection Application.}
%\resizebox{\textwidth}{15mm}{
\begin{tabular}{cl|cl} 
\hline
Parameters                                           & Values                 & Parameters                      & Values                   \\ \hline
\multicolumn{1}{c}{Parallel encoders $Z$} & \multicolumn{1}{c|}{5} & Training epochs $N_{\text{epoch}}$     & \multicolumn{1}{c}{100}  \\ \hline
Embedding size $d_e$                                      &                \multicolumn{1}{c|}{64} & Iterations in each epoch $N_{\text{step}}$ & \multicolumn{1}{c}{5000} \\ \hline
Attention heads $T$                         &                \multicolumn{1}{c|}{4} & Batch size $N_{\text{batch}}$                    &            \multicolumn{1}{c}{256} \\ \hline
 Hidden size of attention $d_h$                       &                \multicolumn{1}{c|}{16} & Learning rate                   &           \multicolumn{1}{c}{$10^{-4}$}  \\ \hline
Hidden size of the component-wise FCN block $d_f$         &                \multicolumn{1}{c|}{256} & Decay factor                    &          \multicolumn{1}{c}{0.1} \\ \hline
Hyperparameter $\lambda$                                  &                \multicolumn{1}{c|}{10}        & Test samples $N_{\text{test}}$        &                          \multicolumn{1}{c}{3000}\\ \hline
\end{tabular} \label{trans_p}
\end{table*}

\begin{figure}[t!]
\begin{center}
  \includegraphics[width=0.45\textwidth]{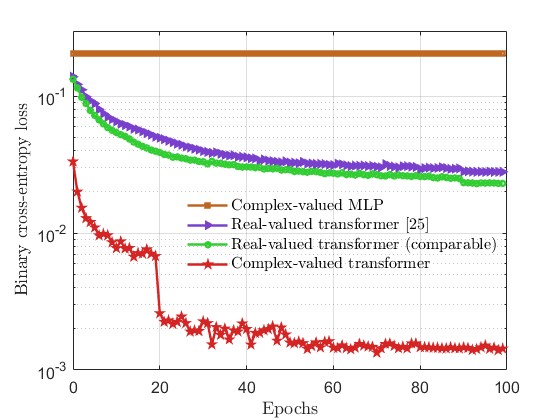}
  \caption{Activity detection validation loss in the training stage.} 
 \label{figAC_loss}
\end{center}
\end{figure}

\textcolor{black}{During the training process, we generate an additional 512 samples for validation and calculate the validation loss using binary cross-entropy. The validation loss versus the number of training epoch is illustrated in Fig.~\ref{figAC_loss}. It can be seen that transformer-based methods can gradually converge during training while complex-valued MLP demonstrates difficulty in converging. Compared to its RVNN counterparts, the complex-valued transformer demonstrates significantly faster convergence, which is markedly improved around the 20th epoch and achieves a substantially lower final validation loss.}

We assess the performance of this application in terms of the probability of missed detection (PM) and the probability of false alarm (PF), which are given by
\begin{align} \label{pm}
\mathrm{PM}=1-\frac{\sum_{n=1}^N \hat{a}_n {a}_n}{\sum_{n=1}^N {a}_n},~~
\mathrm{PF}=\frac{\sum_{n=1}^N \hat{a}_n\left(1-{a}_n\right)}{\sum_{n=1}^N\left(1-{a}_n\right)}, \nonumber
\end{align}
where ${a}_n$ denotes the ground-truth activity, $\hat{a}_{n} \triangleq \mathbb{I}(\textcolor{black}{p_{n}}>\gamma)$ with \textcolor{black}{$p_{n}$} being the inference result returned by the trained model, and $\gamma$ is a threshold between $0$ and $1$. 

\begin{figure*}[ht] 
	\centering
	\subfigure[]{ 
		\label{icc2024_F1}
		\includegraphics[width=0.3\textwidth]{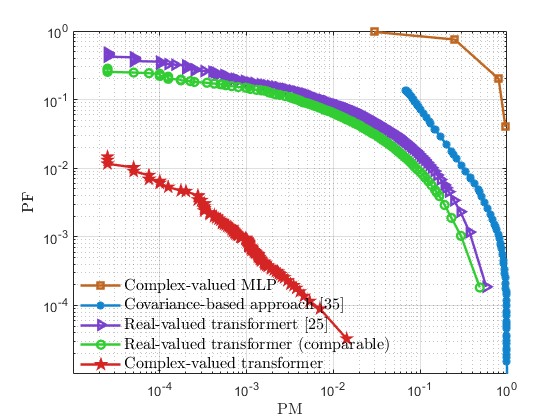}} \hspace{0in}
    \subfigure[]{ 
		\label{icc2024_F2}
      \includegraphics[width=0.3\textwidth]{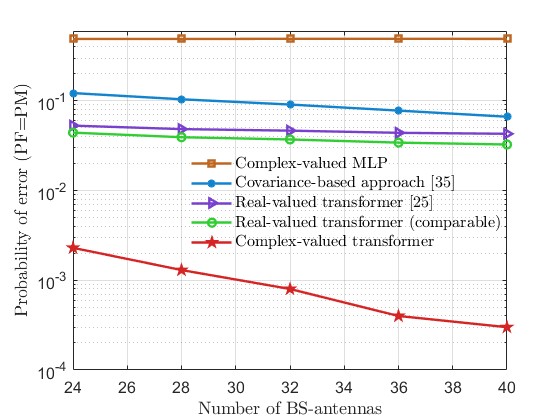}} \hspace{0in}
      \subfigure[]{ 
		\label{fig:AC_reduce}
      \includegraphics[width=0.3\textwidth]{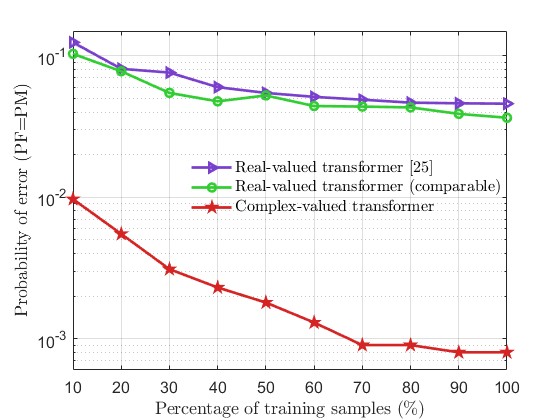}} \hspace{0in}
	\caption{The experimental results of activity detection: (a) performance comparison in terms of PM and PF, (b) generalization to different numbers of BS-antennas, and (c) test performance under limited data training.}  \label{fig:AC}
\end{figure*}

\textcolor{black}{First, the PM-PF curves of various schemes are shown in Fig.~\ref{icc2024_F1}. It can be observed that the complex-valued transformer achieves a much better PM-PF trade-off compared to its real-valued counterparts. In particular, the complex-valued transformer outperforms the original real-valued transformer by a large margin. Even if the real-valued transformer doubles the dimension so that the number of parameters is comparable to that of the complex-valued transformer, the performance is only slightly improved and is still far from that of the complex-valued transformer. This superior performance of complex-valued architecture stems from its ability to capture the correlations while inherently constraining the relationship between the real and imaginary components of the input features. Furthermore, in Fig.~\ref{icc2024_F1}, we include a comparison with a complex-valued MLP network with $4$ layers, and the covariance method, which is an optimization-based method~\cite{haghighatshoar2018}. It can be seen that these two solutions perform much worse than the transformer-based schemes.}

\textcolor{black}{Next, the generalizability ability of different schemes versus the number of BS antennas is shown in Fig.~\ref{icc2024_F2}. Since PM and PF are both important, we select the detection threshold $\gamma$ to make PM = PF, which is denoted as the probability of error. As depicted in Fig.~\ref{icc2024_F2}, the probability of error for the complex-valued transformer is lower than that of the real-valued counterparts and other benchmarks, which demonstrates its generalization ability to accommodate different numbers of BS-antennas.} Similar conclusions can be drawn for varying numbers of devices.

Lastly, we progressively reduce the size of the training set, while maintaining the test set, to further examine how the performance of transformer-based models changes when trained on a limited amount of data. Fig.~\ref{fig:AC_reduce} illustrates that performance generally improves as the number of training samples increases, but the proposed complex-valued transformer shows an unmistakable advantage in terms of detection performance with at least an order of magnitude lower than that of the real-valued counterparts. Notably, even when the two real-valued transformers are trained on a dataset 10 times larger than the complex-valued transformer, their test performance does not come close to that of the complex-valued transformer.

\subsection{Performance on Joint Design of Pilot Sequence, Feedback Quantization, and Precoder}
Unless specified otherwise, we set the number of $\mathrm{BS}$ antennas $N=64$, the number of single-antenna user $K=2$, the uplink carrier frequency to $3.5 \mathrm{GHz}$, the downlink carrier frequency to $3.6 \mathrm{GHz}$. We adopt the well-known Saleh-Valenzuela model \cite{sv-model1,sv-model2} for multipath channel modeling, with the number of channel paths set as $4$ for both uplink and downlink. The path gains are i.i.d. following a circularly symmetric complex Gaussian distribution $\mathcal{C N}(0,1)$. Consider the downlink pilot SNR being $\rho / \sigma^2=10 \mathrm{~dB}$, the corresponding uplink channel is assumed to be perfectly known. \textcolor{black}{ We generate 10,000 samples as training data and another 10,000 samples as validation data for different feedback bits $B = \{1,3,5,10,20,30,40,50,60\}$. We apply standard mini-batch training with stochastic gradient descent and Adam optimizer to maximize the sum rate in an unsupervised way. To evaluate the performance of the achievable sum rate, we generate 10,000 validation samples and test samples for each number of feedback bits $B$ with the same setting of training set. In addition, to test the generalization to different SNR levels, we extend the downlink SNR, $\rho / \sigma^2$, to the range of $0-20$dB at $B=10$ with 10,000 additional test samples for each SNR.}

Furthermore, we use \textcolor{black}{the zero-forcing (ZF) }algorithm to design the downlink precoding matrix for a perfect downlink channel, which serves as an asymptotic upper bound for downlink precoding performance in the high-SNR regime.

\begin{figure}[t!]
    \centering
    \includegraphics[width=0.45\textwidth]{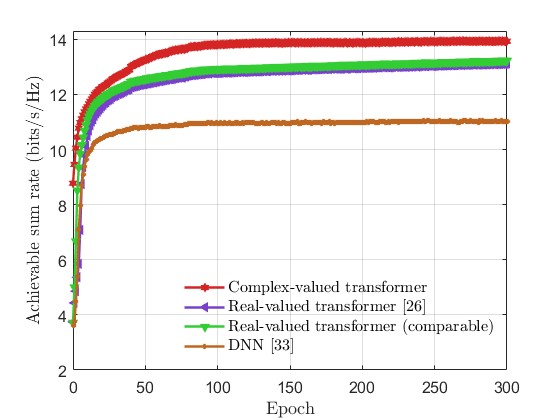}
    \caption{Achievable sum rate of validation set versus training epoch.}
    \label{fig:precoding_loss}
\end{figure}

First, the four methods were subjected to training over 300 epochs. The achievable sum rates of the validation set are illustrated in Fig.~\ref{fig:precoding_loss}. It is evident that the sum rates for all methods initially increase rapidly and then gradually stabilize during the training process. Among these methods, the transformer-based schemes exhibit markedly superior learning capabilities compared to the DNN-based scheme~\cite{sohrabi2021deep}. Notably, the complex-valued transformer consistently maintains a clear advantage throughout the training process.
\begin{figure*}[ht!] 
	\centering
	\subfigure[]{ 
		\label{fig:precoding_B}
		\includegraphics[width=0.3\textwidth]{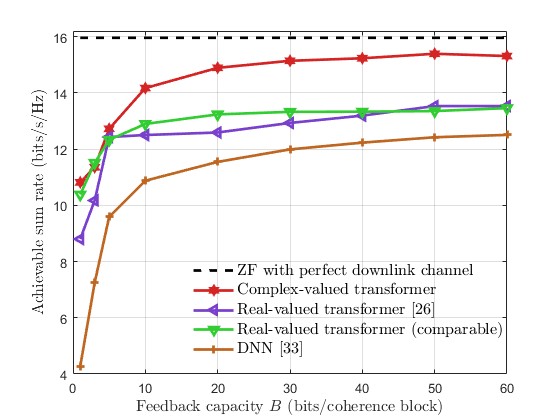}} \hspace{0in}
    \subfigure[]{ 
		\label{fig:precoding_dlsnr}
      \includegraphics[width=0.3\textwidth]{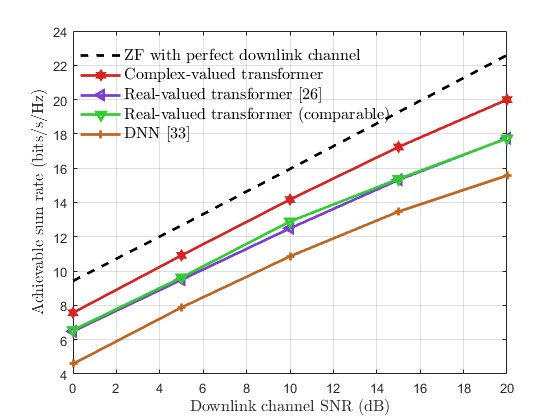}} \hspace{0in}
      \subfigure[]{ 
		\label{fig:precoding_reduce}
      \includegraphics[width=0.3\textwidth]{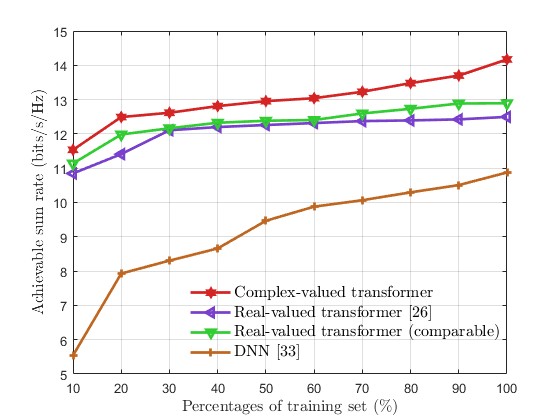}} \hspace{0in}
	\caption{The experimental results of precoding design: (a) Performance comparison of achievable sum rate, (b) generalization to downlink channel SNR, and (c) achievable sum rate under limited data training.}  \label{fig:precoding}
\end{figure*}

Then, we evaluate the four methods on test sets as shown in Fig.~\ref{fig:precoding}.
In Fig.~\ref{fig:precoding_B} the achievable sum rate is illustrated for different numbers of feedback bits. We can observe that the complex-valued transformer-based scheme outperforms the existing real-valued counterpart. For example, when the feedback bits $B=20$, the complex-valued scheme achieves $20\%$ and $30\%$ performance gains over the real-valued transformer~\cite{jiang2023transformer} and the DNN-based scheme~\cite{sohrabi2021deep}, respectively. Moreover, the complex-valued transformer also exhibits an excellent performance compared to the real-valued transformer with a comparable number of parameters, suggesting that increasing the number of parameters in the real-valued transformer does not necessarily result in significant performance improvement. The superior performance of the complex-valued transformer benefits from the intrinsic mechanism of the CVNN. Furthermore, compared to the upper bound, the complex-valued transformer only suffers from a loss of $0.6$ bits/s/Hz (or a relative loss of $3.75\%$).

We also evaluate the performance under different downlink SNRs as shown in Fig.~\ref{fig:precoding_dlsnr}. The end-to-end-based methods are all trained at an SNR of 10 dB and tested on the SNR range of $0-20$dB. The number of feedback information bits $B$ is set to 10. It can be seen that as the SNR increases, the downlink precoding sum rate shows an upward trend. Among them, the performance of the complex-valued transformer-based scheme has the best sum rate performance in the SNR range of $0-20$dB. Moreover, the performance gap between the complex-valued transformer-based scheme and the other end-to-end downlink precoding schemes widens with increasing SNR.

Finally, we further investigate the achievable sum rate under varying percentages of the training set, with the number of feedback information bits $B$ set to 10. As depicted in Fig.~\ref{fig:precoding_reduce},  the achievable sum rate is influenced when the size of the training set is reduced. When compared with existing real-valued schemes, the proposed complex-valued transformer-based scheme achieves superior sum rate performance when the training data is limited. For instance, to reach a sum rate of $13$ bits/s/Hz, the complex-valued scheme requires only half the training data needed by the RVNNs.

\section{Conclusions}
\label{Sec_Conclusion}
In this paper, the potential of the complex-valued transformers in wireless communications was demonstrated. In particular, the basic layers and theoretical analysis of CVNN, core modules in the complex-valued transformer such as the complex-valued attention mechanism, as well as the design of specific layers for mapping complex-valued data to the real-valued domain were thoroughly investigated. The versatility of the complex-valued transformers was illustrated through three applications: channel estimation, activity detection, and downlink precoding, which cover supervised and unsupervised learning, classification and regression, single module design, and end-to-end design. These applications showcased the superior inference performance, excellent generalization, and stable performance of complex-valued transformers compared to their real-valued counterparts. This highlighted the importance of keeping a complex-valued structure in signal and feature as a whole rather than the real and imaginary parts separately. The results of this paper would serve as a guideline for the efficient and practical implementation of complex-valued transformers in wireless communication scenarios.

\bibliographystyle{IEEEtran}
\bibliography{ref}

\end{document}